\documentclass[amssymb,aps,twocolumn,showpacs,nofootinbib, superscriptaddress]{revtex4}
\usepackage{latexsym,amsmath,amssymb,amsfonts,mathbbol,graphicx,color,amsthm,hyperref,enumerate}

\newcommand{\cA}{\mathcal{A}}
\newcommand{\cB}{\mathcal{B}}
\newcommand{\cC}{\mathcal{C}}
\newcommand{\cD}{\mathcal{D}}
\newcommand{\cE}{\mathcal{E}}
\newcommand{\cF}{\mathcal{F}}
\newcommand{\cG}{\mathcal{G}}
\newcommand{\cH}{\mathcal{H}}
\newcommand{\cI}{\mathcal{I}}

\newcommand{\cK}{\mathcal{K}}
\newcommand{\cL}{\mathcal{L}}
\newcommand{\cM}{\mathcal{M}}
\newcommand{\cN}{\mathcal{N}}

\newcommand{\cP}{\mathcal{P}}
\newcommand{\cQ}{\mathcal{Q}}
\newcommand{\cR}{\mathcal{R}}
\newcommand{\cS}{\mathcal{S}}
\newcommand{\cT}{\mathcal{T}}
\newcommand{\cU}{\mathcal{U}}

\newcommand{\ket}[1]{\left| #1 \right\rangle}

\newcommand{\proj}[1]{| #1\rangle\!\langle #1 |}
\newcommand{\Tr}{\mathrm{Tr}}

\newcommand{\hf}{\frac{1}{2}}

\def\Id{1\!\mathrm{l}}
\newcommand{\tr}{\text{tr}}
\newcommand{\supp}{\text{supp}}

\newcommand{\RP}{\hat\cE_\cP}
\newcommand{\Dp}{\Delta_+}
\newcommand{\Dm}{\Delta_-}
\newcommand{\Dpm}{\Delta_\pm}

\newcommand{\Einf}{\cE_\text{inf}}
\newcommand{\EQ}{\cE_R}
\newcommand{\Ne}{{N_\epsilon}}
\newcommand{\Ge}{\cG_\epsilon}
\newcommand{\EP}{\cE_{\cP_0}}
\def\cEP{{\hat\cE}_\cP}
\def\bC{\mathbb{C}}
\def\bR{\mathbb{R}}
\newcommand{\FixE}{\mathrm{Fix}(\cE)}
\newcommand{\RotE}{\mathrm{Rot}(\cE)}
\newcommand{\FixEdag}{\mathrm{Fix}(\cE^\dag)}

\newtheorem{theorem}{Theorem}
\newtheorem{lemma}[theorem]{Lemma}
\newtheorem{corollary}[theorem]{Corollary}

\newtheorem{definition}{Definition}
\newtheorem{lemma1}{Lemma}[theorem]

\newtheorem{example}{Example}
\newtheorem{proposition}{Proposition}
\newtheorem{principle}{Principle}
\newenvironment{numtheorem}[2][Theorem]
{\begin{trivlist}\item[\hskip \labelsep {\bf #1}\hskip \labelsep {\bf #2. }]\it}{\end{trivlist}}
\newenvironment{numlemma}[2][Lemma]
{\begin{trivlist}\item[\hskip \labelsep {\bf #1}\hskip \labelsep {\bf #2. }]\it}{\end{trivlist}}
\newenvironment{numcorollary}[2][Corollary]
{\begin{trivlist}\item[\hskip \labelsep {\bf #1}\hskip \labelsep {\bf #2. }]\it}{\end{trivlist}}

\def\THMPRESCORR{A [convex] code $\cC$ is correctable for $\cE$ if and
only if it is preserved 
by $\cE$.}

\def\THMPRESFIXED{Every maximum preserved code for a CPTP map $\cE$ is
$1$-isometric to the full set of fixed states for some other CPTP map
$\cR\circ\cE$.}

\def\THMFIXEDPTTHM{Let $\cE$ be a CPTP map on $\cB(\cH)$, and
  $\cE^\dagger$ its adjoint.  Let $\FixE$ be the fixed points of $\cE$, and
  $\FixEdag$ the fixed points of $\cE^\dagger$.  Then,
\begin{enumerate}[(i)]
\item Let $\cP_0\subseteq \cH$ be the support of $\FixE$.  Then
$\cP_0$ is an invariant subspace under $\cE$.%: for all $\rho\in\cB(\cP_0)$, $\cE(\rho)\in\cB(\cP_0)$.
\item Let $\EP$ be the restriction of $\cE$ to $\cP_0$, so $\EP\equiv
\Pi_0\circ\cE\circ\Pi_0$, where $\Pi_0$ projects onto $\cP_0$.  Then
the fixed points of $\EP^\dagger$ form a matrix algebra $\cA$.
\item $\FixE$ is a distortion of $\cA$.
\item $\FixEdag$ is a 1:1 extension of $\cA$ from $\cP_0$ to $\cH$.
That is, for each $X\in\cA$, there exists precisely one
$X'\in\FixEdag$ so that $X=\Pi(X')=P_0 X' P_0$.
\end{enumerate}}

\def\LEMNSFIXED{Every noiseless code $\cC$ for $\cE$ is isometric to a
set of states that are fixed points of $\cE$.}  

\def\CORNSMAX{Every maximum noiseless code for a channel $\cE$ is
isometric to the full fixed-point set of $\cE$.}

\def\LEMNSSTRUCT{Let $\cE:\cB(\cH)\to\cB(\cH)$ be a CP map with a
full-rank fixed point, whose fixed points induce (see Theorem
\ref{thm:FixedPtThm}) the decomposition
\begin{equation*}
\cH=\bigoplus_k({A_k}\otimes {B_k}).
\end{equation*}
Then $\cC$ is a [convex] maximum noiseless code for $\cE$ if and only
if $\cC$ comprises all states of the following form
\begin{equation}
\rho = \sum_k{p_k\rho_{A_k}\otimes \tau_{k}},
\end{equation}
where the $\rho_{A_k}$ are arbitrary states on ${A_k}$ and each
$\tau_{k}$ is a fixed ({\em i.e.}, the same for all $\rho$) state on
${B_k}$.}

\def\LEMFINDINGIPSISHARD{The problem of finding the largest preserved
IPS for an arbitrary channel $\cE:\cB(\cH_d)\to\cB(\cH_{d^2})$ that
maps a $d$-dimensional system to a $d^2$-dimensional system is at
least as hard as the NP-complete problem \textbf{MAX-CLIQUE}.}

\def\LEMUNROT{If $\cC$ is a maximum unitarily noiseless code for a CP
map $\cE$, then $\cC$ is isometric to the set of all (positive
trace-1) states in the span of the rotating points of $\cE$.  In other
words, there exists a map $\Einf$ such that $\Vert
p~\Einf(\rho)-(1-p)\Einf(\sigma)\Vert_1=\Vert
p\rho-(1-p)\sigma\Vert_1$ for any $\rho,\sigma\in\cC$, $p\in[0,1]$,
and $\Einf(\rho)$ and $\Einf(\sigma)$ are in the span of the rotating
points of $\cE$.}

\begin{document}

\title{Information preserving structures: \\ A general framework for
quantum zero-error information}

\author{Robin Blume-Kohout} \email{robin@blumekohout.com}
\affiliation{Perimeter Institute for Theoretical Physics, 31 Caroline
Street North, Waterloo, ON N2L2Y5, Canada}

\author{Hui Khoon Ng} \email{nghk@theory.caltech.edu}
\affiliation{Institute for Quantum Information, California Institute
of Technology, Pasadena, CA 91125, USA}

\author{David Poulin} \email{david.poulin@usherbrooke.ca}
\affiliation{D\'epartment de Physique, Universit\'e de Sherbrooke,
Qu\'ebec J1K 2R1, Canada}

\author{Lorenza Viola} \email{lorenza.viola@dartmouth.edu}
\affiliation{\mbox{Department of Physics and Astronomy, Dartmouth College,
6127 Wilder Laboratory, Hanover, NH 03755, USA}}

\date{\today}

\begin{abstract}

Quantum systems carry information.  Quantum theory supports at least
two distinct kinds of information (classical and quantum), and a
variety of different ways to encode and preserve information in
physical systems.  A system's ability to carry information is
constrained and defined by the noise in its dynamics.  This paper
introduces an operational framework, using \emph{information-preserving
structures} to classify all the kinds of information that can be
perfectly (i.e., with zero error) preserved by quantum dynamics.  We
prove that every perfectly preserved code has the same structure as a matrix
algebra, and that preserved information can always be corrected.  We also
classify distinct operational criteria for preservation ({\em e.g.},
``noiseless'', ``unitarily correctible'', {\em etc.}) and introduce
two new and natural criteria for \emph{measurement-stabilized} and
\emph{unconditionally preserved} codes.  Finally, for several of these
operational critera, we present efficient [polynomial in
%$\mathrm{dim}(\mathcal{H})$
the state-space dimension] algorithms to find all of a channel's
information-preserving structures.

%We introduce a general operational characterization of information-preserving structures in quantum theory -- encompassing noiseless subsystems, decoherence-free subspaces, pointer bases and error-correcting codes -- by demonstrating that they are isometric to fixed points of unital quantum processes. Using this, we show that every information-preserving structure is a matrix algebra. We further establish a structure theorem for fixed states and observables of an arbitrary process, which unifies the Schr{\"o}dinger and Heisenberg pictures, and provides an efficient algorithm for finding all noiseless and unitarily noiseless codes of the process.

%\todo{
%\begin{enumerate}
%\item Does some discussion of prior work on Markovian noise need to be inserted? (Lorenza?)
%\item Some inspired transition language needs to be put in Section I (see "todo" mark there)
%\item We need to consider \& approve/reject Def. 6 on pg. 10.
%\item We need to proofreas the abstract.
%\item Need to proofread the proofs section, esp. Lemma 6.
%\item Need to compare text of theorems/lemmas in main text and Appendix B -- if different, reconcile.
%\end{enumerate}
%}
\end{abstract}

%%LV: Note I switched pacs and fixed one of them 

\maketitle

%******************************************************************
%******************************************************************
\section{Introduction}
\label{sec:Intro}

Physical systems can be used to store, transmit, and transform
information.  Different systems can carry different kinds of
information; classical systems carry classical information, while
quantum mechanical systems can carry quantum information.  The
system's dynamics also affect the kind of information that it
carries.  For example, decoherence \cite{Zurek} can restrict a quantum
system to carry only classical information (or none at all).  This
suggests that perhaps a quantum system's dynamics can select other
kinds of information, neither quantum nor classical, but something in
between.  The central result of this paper is an exhaustive
classification of exactly what kinds of information can be selected in
this way.

Preservation of information in physical systems is important in
several contexts.  In communication theory, information originates
with a sender (``Alice'') who actively conspires with a receiver
(``Bob'') to transfer it over a communication channel.  Computational
devices require memory registers that can store information in the
face of repeated noise.  Experimental and observational sciences
require, in a more or less explicit way, the transmission of
information from a passive system of interest (perhaps a distant
galaxy, or a nanoscale device), through a chain of ancillary systems,
to an observer.  In each case, achieving the desired transformation
requires first that the information be \emph{preserved} by a noisy
dynamical process or ``channel'' -- yet, each operational scenario
poses a subtly different notion of ``preserved''.

In this paper we develop a theory that covers all these situations in
a unified framework.  We start by establishing a general setting for
information (and its preservation), using \emph{codes} (Section
\ref{sec:PresInfo}).  We state a minimal necessary condition for
information preservation, then prove that it is also sufficient (in a
particular strong sense), deriving a powerful structure theorem for
preserved codes (Section \ref{sec:Struct}).  On this foundation, we
build a hierarchy of different \emph{operational} criteria for
preservation (Section \ref{sec:Types}).  Stricter criteria correspond
to additional operational constraints -- {\em e.g.}, that information
persist for more than one application of the noise.  On the one hand,
some of these criteria allow us to make natural contact with
previously studied approaches to information preservation -- including
pointer states \cite{Zurek}, decoherence-free subspaces \cite{DFS} and
noiseless subsystems \cite{KnillPRL00,NSqubit,NS}, and quantum error
correcting codes \cite{QEC} -- while also proposing a couple of new
ones, related to what we call ``measurement-stabilized'' and
``unconditionally preserved'' codes.  On the other hand, our main
contribution is to gather them all into a single framework using
\emph{information-preserving structures} (IPSs).
%% LV: Note that in PRL we used IPS for plural -- even though I always 
%% liked the current version better so I'd keep it.
IPSs classify the kinds of information that dynamical processes can
preserve.  In particular, we focus here on \emph{perfect IPS},
corresponding to \emph{zero-error} information.  Finally, we consider
how to find these structures for a given noisy process (Section
\ref{sec:App}).  It is NP-hard to find a channel's largest correctible
IPS, but for stricter preservation criteria it can be much easier.  We
provide efficient and exhaustive algorithms to find noiseless,
unitarily noiseless, and unconditionally preserved IPSs.

Our IPS framework establishes an explicit and rigorous connection
between perfectly preserved information and fixed points of channels.
By focusing on fixed points (see also \cite{Beny}), rather than on the
noise commutant, it provides a first step toward understanding
\emph{approximate IPS}, making contact with stability results for
decoherence-free encodings under symmetry-breaking perturbations
\cite{BaconPRA99}, and with approximate QEC
\cite{aQEC,NgBK,Ticozzi2010}.  
%% LV: The previous/above sentence was a little too close to the one in PRL...
Our structure theorem for the fixed points of completely positive maps
extends previous results that apply only to unital processes
\cite{AGG02a,KribsPEMS03}, or processes with a full-rank fixed state
\cite{Frigerio}.  Our algorithm for finding noiseless and unitarily
noiseless codes improves on algorithms that are inefficient ({\em
e.g.}, Refs. \onlinecite{ZurekPTP93,ChoiPRL06}), or otherwise restricted to
purely noiseless information \cite{Knill06a} or unital channels
\cite{KS06a}.

Early aspects of this work appeared in Ref. \onlinecite{BNPV}.  Here,
we provide more results, full proofs, and detailed discussion.

%******************************************************************
%******************************************************************
\section{Preserved information}
\label{sec:PresInfo}

``What kinds of information can a quantum dynamical process
preserve?'' is a technical question, but one that requires a firm
conceptual foundation.  This section aims to provide one.  We begin
with an operational definition of ``information,'' then apply it to
quantum theory.  We use well-known results on the accuracy with which
quantum states can be distinguished to establish a mathematical
framework in which this central question can be answered.

``Information'' has a variety of meanings.  Any crisp definition will
inevitably run afoul of some alternative usage.  Throughout
\emph{this} paper, we will follow this basic operational definition:
\begin{principle}\label{def:information}
Information is a resource, embodied in a physical system, that can be
used to answer a question.
\end{principle}

A physical system $\cS$ can carry information.  If one party (Alice)
sends it to another (Bob), then the recipient can use it to answer a
question.  More precisely, possession of $\cS$ gives Bob a higher
probability of guessing the correct answer.  However, if $\cS$ evolves
during transmission -- {\em i.e.}, it undergoes a dynamical map $\cE$
-- then some information might be lost.  As a result, $\cE(\cS)$ may
be less useful than $\cS$.  It is not yet clear how to determine
whether information is ``preserved'', but two principles seem
self-evident:
\begin{principle}
\label{def:preservation_trivial}
If nothing happens to a system, then all the information in it is
preserved.
\end{principle}
\begin{principle}
\label{def:nonpreservation}
If a system evolves as $\cS\to\cE(\cS)$, and $\cE(\cS)$ is strictly
less useful than $\cS$ in answering some question, then some
information in $\cS$ was \textbf{not} preserved.
\end{principle}

These simple criteria bracket the (as-yet undefined) notion of
preservation -- of \emph{all} the information in a system.  But
information can be \emph{encoded} into one part of a
system.  Such information may be preserved even if other parts are
damaged or destroyed.  To properly represent this notion, we appeal to
another self-evident principle:
\begin{principle}
\label{def:triviality}
If some property or parameter of a system is already known to all
parties ({\em e.g.} Alice and Bob), then it carries no useful
information.
\end{principle}

For example, if a quantum system $\cS$ is \emph{known} to be in the
state $\proj{\psi}$, by all parties, then nothing is gained
by transmitting it.  Since a known property of $\cS$ carries no
information, disturbing it has no effect on the information embodied
in the system.  So, we can represent the sequestering of information
in a very general way by stating a promise or precondition, which
\emph{guarantees} certain properties of $\cS$.  Those properties,
being already known, carry no useful information.  Information carried
by $\cS$ \emph{conditional} on the promise can be preserved, even if
other properties (constrained by the promise) are disturbed.

Mathematically, a precondition on $\cS$ is a restriction of its state,
to some (arbitrary) subset.  We call such a set a \emph{code}.
\begin{definition}
\label{def:code}
A code $\cC$ for a system $\cS$ is an arbitrary subset of the system's
state space.
\end{definition}

Codes carry information.  Each system $\cS$ has a natural ``maximum
code'' containing all its possible states.  Smaller codes for that
system carry strictly less information -- but may be preserved even
when the system's maximum code is not.  A code that is a strict subset
of another preserved code is uninteresting, so we will focus on
\emph{maximal} preserved codes.
\begin{definition}
\label{def:maximal}
A preserved code $\cC$ is \textbf{maximal} if there exists no
preserved $\cC_{\mathrm{big}} \supset \cC$.  That is, if adding any
other state would render $\cC$ unpreserved.
\end{definition}

We can narrow our focus even more.  If $\cS$ has two preserved codes,
$\cC_{\mathrm{big}}$ and $\cC_{\mathrm{small}}$, where
$\cC_{\mathrm{big}}$ is strictly ``bigger'' than
$\cC_{\mathrm{small}}$, then we are not interested in
$\cC_{\mathrm{small}}$.  $\cC_{\mathrm{big}}$ is ``bigger'' than
$\cC_{\mathrm{small}}$ if it has a proper subset that is identical or
isomorphic to $\cC_{\mathrm{small}}$.  We can make this rigorous, but
only by borrowing a technical definition from the next section (see
Definition \ref{def:isometric}):
\begin{definition}
A preserved code $\cC$ is \textbf{maximum} if and only if there is no
preserved $\cC_{\mathrm{big}}$ such that $\cC$ is isometric to a
strict subset $\cC_{\mathrm{small}} \subset \cC_{\mathrm{big}}$.
\end{definition}
We will generally restrict our attention to maximum
codes\footnote{Graph theorists may recognize this terminology.
Maximal and maximum codes have the same relationship as maximal and
maximum cliques, or independent sets.  Note, however, that unlike a
graph, a channel need \emph{not} have a unique maximum code.  If a
channel preserves \emph{either} a quantum bit \emph{or} a classical
trit, they are incomparable -- neither is bigger than the other.}.  We
need a precise definition of a ``preserved'' code.  We begin by
adapting Principles \ref{def:preservation_trivial} and
{\ref{def:nonpreservation} to codes:
\begin{principle}
\label{def:codepreservation}
The information in a code $\cC$ is preserved by a dynamical map $\cE$
\textbf{if} $\cE$ leaves every state in $\cC$ unchanged.
\end{principle}
\begin{principle}
\label{def:codenonpreservation}
The information in a code $\cC$ is preserved by a dynamical map $\cE$
\textbf{only if} $\cE(\cC)$ is as useful as $\cC$ for answering any
question.
\end{principle}

These are sufficient and necessary (respectively) \emph{operational}
conditions for preservation.  Principle \ref{def:codenonpreservation}
seems much weaker than \ref{def:codepreservation} -- but we will show
that it is actually not.  If Principle \ref{def:codenonpreservation}
is satisfied, then there is a physically implementable \emph{recovery
operation} that restores every code state.  The ability to perform
this recovery is a resource -- a reasonable one, but a nontrivial one.
We will also consider several weaker resources ({\em e.g.},
restrictions on what recovery operations can be implemented), and the
corresponding stronger notions of preservation, in Section
\ref{sec:Types}.

This concludes the ``philosophical'' part of our framework, and in
what follows we will build on these foundations to establish technical
results.  Two final points deserve mention, however:

(i) Identifying ``information'' with codes (arbitrary sets of states)
is intended to be a very general paradigm.  A system's state, by
definition, specifies everything that can be known about that system.
Every question that can be answered using $\cS$ boils down to a
question about the state of $\cS$, and variations in that state
(restricted to some particular code) encode information.  If there are
exceptions to this rule -- that is, notions of information, consistent
with Principle \ref{def:information}, that cannot be represented using
codes -- then we are not aware of them\footnote{A simple and important
example is entanglement between $\cS$ and a reference system $\cR$.
Though not \emph{explicitly} mentioned, entanglement is easy to
characterize in our setting.  If $\cS$ and $\cR$ are maximally
entangled, then $\cS$ can be post-selectively prepared in any pure
state $\proj{\psi}$ by projecting $\cR$ into some $\proj{\psi'}$.
Entanglement is preserved if and only if the code containing
\emph{all} of these conditional states is preserved.}.  An extended
discussion can be found in Appendix \ref{sec:CodeApology}).

(ii) Our definition of ``information'' may not appear congruent with
Shannon's theory of communication \cite{CoverBook91,ShannonBook49}.
In fact, it is quite compatible.  There are, however, some subtle
differences: as mentioned, we focus on zero-error information;
furthermore, we consider a \emph{single} use of a communication
channel, rather than $N$ uses with $N\to\infty$.  An extended
discussion can be found in Appendix \ref{sec:ShannonTheory}.

%********************************************************************
\subsection{Systems, states, codes, and channels in quantum theory}

So far, we have used a language consistent with a broad range of
physical theories.  Let us now specialize to quantum theory.  States
of quantum systems are represented by density operators $\rho$, which
are positive trace-$1$ operators on the system's Hilbert space $\cH$.
Quantum dynamical maps (also known as \emph{channels}) are described
by completely positive (CP), trace-preserving (TP) linear maps on
density operators.  A CPTP map $\cE$ can be represented in two
equivalent ways.  In one formulation, the initial system $\cS_A$ comes
into contact with an \emph{uncorrelated} environment $E_0$, they
evolve unitarily, and then some part $E_f$ of this joint system is
discarded\footnote{A technical note is in order here.  If the
environment $E_0$ is initially correlated with the input system
$\cS_A$, then the resulting dynamics is generally \emph{not} CP, and
so initial decorrelation is a common assumption in the theory of open
quantum systems.  For our purposes, it is more than just an
assumption.  If $\cS_A$ is initially correlated with its environment,
then the latter contains information about $\cS_A$.  The system and
its environment \emph{together} may contain more information about
$\cS_A$ than does $\cS_A$ itself!  In the course of the ensuing
interaction, that information may flow back into the system.  It is
impossible (ill-defined, even) to say whether information in $\cS_A$
has been preserved in such a case, for it may have been replaced with
information initially residing in $E_0$.  Such an interaction is not,
in any sense, ``noise''.}, yielding a reduced state for the final
system $\cS_B$:
\begin{equation}
\rho_B = \cE(\rho_A) =
\Tr_{E_f}\left[U\left(\rho_A\otimes\rho_{E_0}\right)U^\dagger\right].
\end{equation}
The other representation of a CP-map is called the operator-sum
representation:
\begin{equation}
\rho_B = \cE(\rho_A) = \sum_i{K_i \rho_A K_i^\dagger},
\end{equation}
where the \emph{Kraus operators} $\{K_i\}$ satisfy $\sum_i{K_i^\dagger
K_i} = \Id$.  This representation is mathematically simpler but less
physically intuitive (for a complete treatment of CP maps, see
Refs. \onlinecite{KrausBook83,NielsenBook00}).  Note that in either
representation, $\cS_A$ and $\cS_B$ may be different systems, with
different Hilbert spaces.  However, the special case where they are
the same is very important -- for instance, all continuous-time
processes are described by such maps -- and we will often implicitly
assume it, dropping $A$ and $B$ subscripts and relying on context to
illustrate whether ``$\cS$'' refers to the channel's input or its
output.

Codes for quantum systems are sets of quantum states, e.g. $\cC =
\{\rho_1\ldots\rho_k\}$.  The code represents a promise that the
system will be prepared in some $\rho\in\cC$.  Each distinct code
represents a potentially distinct kind of information.  Note, however,
that we are not introducing an infinite proliferation of fundamentally
different ``kinds'' of information, nor are we suggesting that a qubit
carries fundamentally different information from a qutrit: Systems
with isomorphic state spaces carry the same kind of information.  $N$
qutrits equal $\log_23$ qubits, so they carry the same kind of
information, but more of it.  The important dividing line is between
systems that have no asymptotic equivalence, like a qubit and a
classical bit\footnote{Two systems $\cS_A$ and $\cS_B$ have an
asymptotic equivalence if there is a constant $R$ such that for all
$\epsilon >0$ and $N \rightarrow \infty$, (i) $N(R- \epsilon)$ copies
of $\cS_A$ is strictly less powerful that $N$ copies of $\cS_B$, and
(ii) $N(R+ \epsilon)$ copies of $\cS_A$ is strictly more powerful that
$N$ copies of $\cS_B$. Thus, any two finite non-trivial quantum
systems have an asymptotic equivalence in this sense.}.

Now that we have a well-defined mathematical theory, we need a
mathematical definition of preservation.  Principle
\ref{def:codenonpreservation} uses the very general idea of
``questions.''  A simple and well-defined set of questions turns out
to be sufficient: ``Was the system prepared in state $\rho$ or state
$\sigma$?''  Here, $\rho$ and $\sigma$ are states in the code $\cC$.
In general, these questions cannot be answered with certainty, for
most pairs of states are not perfectly distinguishable.  But if Bob
cannot distinguish them as well as Alice, then information has been
lost.  Of course, there may well be many other questions that
\emph{could} be asked, but it turns out that if \emph{these}
well-defined questions are all preserved, then the code can be
corrected (and therefore \emph{every} question must be preserved!)

\begin{example}
\label{ex:Codes}
Suppose that $\cS$ is a quantum bit.  If its dynamics are noiseless,
then every state passes unchanged through the channel.  We can
describe the preserved information in terms of a code
$\cC_{\mathrm{qubit}}$ that contains \emph{all} the possible states
for a qubit.  Now, suppose $\cS$ experiences a dephasing channel,
which transforms an arbitrary superposition of the computational
states $\ket{0}$ and $\ket{1}$ into a mixture,
$$\cE: \alpha\ket{0}+\beta\ket{1} \longrightarrow |\alpha|^2\proj{0} +
|\beta|^2\proj{1},$$ and which maps the Bloch sphere into itself like
this:
\begin{center}\includegraphics[width=25mm]{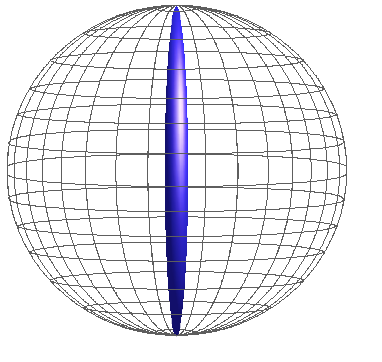}
\end{center}

The code $\cC_{\mathrm{qubit}}$ is no longer preserved.  Because the
two states $\ket{\pm} = \frac{\ket{0}\pm\ket{1}}{\sqrt2}$ are both
mapped to $\rho_B = \frac12{\Id}$, Bob cannot answer the
question ``Was $\cS$ prepared in $\ket{+}$ or $\ket{-}$?''  However,
the more restricted code $\cC_{\mathrm{cbit}} = \{\proj{0},\proj{1}\}$
\emph{is} preserved, for Bob can distinguish between these states just
as well as Alice.  The preserved code describes a different kind of
information: one classical bit.
\end{example}

Here are some familiar examples of preserved information, represented
as codes.

\begin{example}
\label{ex:PointerBases}
A \textbf{pointer basis} comprises a set of mutually orthogonal
``pointer states'' $\{\ket{\psi_1}\ldots\ket{\psi_N}\}$ that are
unaffected (or ``least affected'') by noise -- as originally
introduced in the study of quantum measurement and decoherence
\cite{Zurek}.  A pointer basis can be described by the code containing
all the pointer states (PSs) $\proj{\psi_k}$ and their convex
combinations.  Classical information is stored in the index $k$, but
not quantum information, because superpositions are not preserved, and
thus cannot be included in the code.  PSs are preserved in the
strongest possible sense: Every state in the code is a fixed point of
$\cE$.
\end{example}

\begin{example}
\label{ex:DFS}
A \textbf{decoherence-free subspace} (DFS) is an entire subspace of
the system's Hilbert space, $\cP\subseteq\cH$, which is invariant
under the noise \cite{DFS} (see also Zurek's prior discussion of
``pointer subspaces'' \cite{ZurekPRD82}).  The corresponding code $\cC$ contains
every density operator supported on $\cP$.  Since $\cC$ includes
superpositions of any given basis for $\cP$, a DFS preserves quantum
information, and can in principle support encoded quantum computation.
Like pointer bases, DFSs are preserved in the strongest sense
(although, especially in the context of Markovian dynamics, the
definition is commonly relaxed to allow unitary evolution, see also
\cite{ShabaniPRA05,TicozziTAC}).
\end{example}

\begin{example}
\label{ex:NS}
A \textbf{noiseless subsystem} (NS) shares with a DFS the property
that it can store quantum information.  Unlike a DFS, an NS can exist
even if no pure state in $\cH$ is invariant.  According to the original
definition \cite{KnillPRL00,NSqubit}, it suffices that the noise has a
trivial action on a ``factor'' of $\cH$.  That is, $\cS$ supports an
NS if there exists a subspace $\cH_{AB}\subseteq \cH$ that can be
factored as $\cH_{AB}=\cH_A\otimes\cH_B$, so that for every pair
of states $\rho_A$, $\rho_B$ supported on $\cH_A$, $\cH_B$,
respectively,
\begin{equation}
\cE\left(\rho_A\otimes\rho_B\right) = \rho_A \otimes \rho'_B,
\label{defNS0}
\end{equation}
for some state $\rho'_B$ on $\cH_B$. Thus, the restriction of $\cE$ to
$\cH_{AB}$ obeys 
\begin{equation}
\cE = {\Id}_A \otimes \cE_B ,
\label{defNS1}
\end{equation}
for some CPTP map on $\cH_B$.  Since, for every state $\rho_{AB}$
supported on $\cH_{AB}$,
\begin{equation}
\label{defNS2}
\tr_B\cE(\rho_{AB})=\tr_B\rho_{AB},
\end{equation}
it is clear that quantum information is preserved in the reduced state
of subsystem $A$.  However, it is not immediately obvious that (as in Examples
\ref{ex:PointerBases}-\ref{ex:DFS}) there is a corresponding fixed code for
$\cS$. In fact, the existence of such a code follows from
Eq. (\ref{defNS1}) and the fact that every channel $\cE_B$ has at
least one fixed point $\tau_B$ \cite{GranasBook03}.
%Since Equation \ref{defNS} must apply to product states
%$\rho_{AB} = \proj{\psi_A}\otimes\rho_B$,
%\begin{equation}
%\cE\left(\proj{\psi_A}\otimes\rho_B\right) = \proj{\psi_A}\otimes\rho'_B,
%\end{equation}
Thus, the code $\cC_{\mathrm{NS}} = \{\rho_A\otimes\tau_B,\ \forall
\rho_A\}$, where $\rho_A$ is arbitrary on $\cH_A$, but $\tau_B$ is
fixed, is invariant under $\cE$.
\end{example}

\begin{example}
\label{ex:QECC}
A \textbf{quantum error correcting code} (QECC) \cite{QEC,KLP05a} also
preserves quantum information, but according to a weaker criterion.  A
QECC is a subspace $\cP$ for which there exists a physical recovery
operation $\cR$ so that $(\cR\circ\cE)(\ket\psi)=\ket\psi$ for all
$\ket\psi\in\cP$.  As with a DFS, the corresponding ``correctable
code'' contains all states supported on $\cP$.  Unlike the previous
examples, this code is \emph{not} fixed.  However, it is clearly
preserved, because $\cP$ can be turned into a DFS by applying $\cR$.
An ``operator QECC'' \cite{OQECC} is an NS for $\cR\circ\cE$.
Another variant stipulates active intervention {\em before} the
noise occurs \cite{KnillPRL00}, in which case the code is ``protectable''
rather than correctable \cite{Knill06a}.  While protectable
codes will not be further discussed in the present work, the notions
of protectability and correctability are not fundamentally different
and may, to a large extent, be viewed as ``dual'' to one another, as 
elucidated in \cite{Ticozzi2010}.
\end{example}

The above examples are not exhaustive, but they illustrate the
diversity of criteria for ``preserved'' information.  Each example is
specified by a different algebraic condition, dictated either by
operational constraints or by its relevance to the task at hand.  We hope
that unifying them will bring clarity to experimental implementations
of these ideas \cite{DFSexp,NSexp,QECexp}.

The key point of our framework, though, is to explore \emph{beyond}
these well-known examples.  In particular, all the situations
illustrated above can be described intuitively as ``quantum
information'' or ``classical information.''  What we would like to
know is whether more exotic codes are possible -- whether some weird
channel can preserve a form of information that is entirely unlike a
pointer basis, NS, or QECC.  We need a rigorous criterion for
preservation of codes, based on Principles \ref{def:codepreservation}
and \ref{def:codenonpreservation}.  Principle
\ref{def:codepreservation} is straightforward, but Principle
\ref{def:codenonpreservation} refers to \emph{any} operational task.
Our strategy will be to identify one particular task --
\emph{distinguishing} between code states.  Because we focus on just
one task, we will obtain a \emph{necessary} condition.  Having done
so, our next challenge will be to bring these conditions together.

%********************************************************************
\subsection{Single-shot distinguishability, Helstrom's theorem, 
and the 1-norm}

%Condition \ref{NecCond1} is phrased operationally.  To prove anything further, we need a mathematical measure of distinguishability -- and, in selecting a particular measure, we need to avoid weakening the condition.  As with the claim that Condition \ref{NecCond1} is actually sufficient, the sufficiency of the measure we are about to define will only be apparent after the fact.  

Suppose that Bob has access to a single copy of the system $\cS$, and
he wishes to guess correctly whether it was prepared in state $\rho$
or state $\sigma$ (both of which are in $\cC$).  He seeks to maximize
the probability that his guess is correct, and he knows that the prior
probabilities of $\rho$ and $\sigma$ are (respectively) $p$ and
$(1-p)$.  He can measure $\cS$ to help him decide, and the optimal
course of action is determined by Helstrom's theorem
\cite{HelstromBook76}:\\

\noindent\textbf{Helstrom's Theorem.}\textit{ Suppose a quantum system
$\cS$ was prepared in either in state $\rho$ or in state $\sigma$,
with respective probabilities $p$ and $(1-p)$.  The highest
probability of guessing correctly which was prepared is obtained by
measuring the Hermitian operator $\Delta_p = p\rho - (1-p)\sigma$,
then guessing ``$\rho$'' upon obtaining a result corresponding to a
positive eigenvalue and ``$\sigma$'' in the case of a negative
eigenvalue.  If a zero eigenvalue is obtained, either guess is equally
good. The success probability is given by $P_H(\rho,\sigma;p)=\frac
12(1 + \| \Delta_p \|_1)$, where $\Vert\cdot\Vert_1$ refers to the
1-norm, $\Vert A\Vert_1\equiv \tr\,\vert A\vert=\tr\sqrt{A^\dagger
A}$.
\label{thm:Helstrom}
}\\

The success probability $P_H$ is a measure of the {\em
distinguishability} between $\rho$ and $\sigma$.  It is non-increasing
under any CPTP map, because the 1-norm is contractive under CPTP maps
\cite{Contractivity}.  So, in order for $\{\cE(\rho),\cE(\sigma)\}$ to
be as distinguishable as $\{\rho,\sigma\}$, we require that for every
prior probability $p$, the Helstrom strategy yields the same success
probability for distinguishing $\rho$ from $\sigma$ as for
distinguishing $\cE(\rho)$ from $\cE(\sigma)$:
$$P_H(\cE(\rho),\cE(\sigma);p)=P_H(\rho,\sigma;p).$$ If Bob needs to
distinguish between two \emph{sets} of states, $\{\rho_k\}$ and
$\{\sigma_k\}$, he assigns prior probabilities $\{p_k\}$ and $\{s_k\}$
to the $\{\rho_k\}$ and $\{\sigma_k\}$, respectively.  Then his task
is to distinguish
$$\rho = \frac{1}{\sum_k{p_k}}\sum_k{p_k\rho_k}$$
from
$$\sigma = \frac{1}{\sum_k{s_k}}\sum_k{s_k\sigma_k},$$ 
\noindent 
where the prior probabilities of $\rho$ and $\sigma$ are,
respectively, $p = \sum_k{p_k}$ and $1-p$.

This measure of distinguishability is, in fact, a metric on the space
of linear operators.  Its preservation implies a kind of rigid
equivalence, which we make precise with the following definition:
\begin{definition}\label{def:isometric}
Two codes $\cC_1$ and $\cC_2$ are \textbf{1-isometric} (or just ``isometric'')
to each other if and only if there exists a linear 1:1 mapping $f:\cC_1\to\cC_2$ such that,
for all $\rho,\sigma$ in the convex closure of $\cC_1$ and all
$p\in[0,1]$, 
$$\Vert pf(\rho)-(1-p)f(\sigma)\Vert_1=\Vert
p\rho-(1-p)\sigma\Vert_1.$$
\end{definition}

\begin{definition}\label{def:isometric2}
A code $\cC$ is \textbf{1-isometric} (or just ``isometric'') for a
CPTP process $\cE$ only if $\cC$ is isometric to $\cE(\cC)$.
\end{definition}

So, if a code is isometric for a given map $\cE$, then $\Vert
p\cE(\rho)-(1-p)\cE(\sigma)\Vert_1=\Vert p\rho-(1-p)\sigma\Vert_1$ for
all $\rho,\sigma$ in the convex closure of $\cC$ and $p\in[0,1]$. A
stronger characterization is given by the following:

\begin{definition}
\label{FixedDef}
A code $\cC$ is \textbf{fixed} by a CPTP channel $\cE$ if and only if
$\cE(\rho) = \rho$ for all $\rho\in\cC$.
\end{definition}

%**********************************************************************
\subsection{Criteria for preservation}

We are now in a position to state Principle \ref{def:codepreservation}
more precisely:\\

%\begin{condition}\label{SufCond}
\noindent\textbf{Strong Condition for Preservation.}  \textit{A
\textbf{sufficient} condition for $\cC$ to be preserved by $\cE$ is
that $\cC$ be fixed by $\cE$.}\\
%\end{condition}

The Strong Condition is obviously sufficient, but (as demonstrated by
error correcting codes) it is not \emph{necessary} for preservation.
Principle \ref{def:codenonpreservation} implies a host of necessary
conditions -- one for every operational task.  We choose one in
particular: We demand that $\cE(\rho)$ and $\cE(\sigma)$ be just as
distinguishable\footnote{Note that $\rho$ and $\sigma$ need not be
perfectly distinguishable to start with.  A QECC contains
non-orthogonal states that cannot be perfectly distinguished, but they
can be distinguished just as well after $\cE$ as before.} as $\rho$
and $\sigma$.  We also require that questions like ``Was $\cS$
prepared in one of the states $\{\rho_1,\rho_2,\rho_3\ldots\}$, or in
one of the states $\{\sigma_1,\sigma_2,\sigma_3\ldots\}$?'' should be
preserved as well, so convex combinations of code states should
maintain their pairwise distinguishability.  There is nothing
inherently special about this particular operational task, except that
it produces a useful and convenient mathematical condition: \\

%\begin{condition}\label{NecCond1}
\noindent\textbf{Weak Condition for Preservation.}  \textit{A
\textbf{necessary} condition for $\cC$ to be preserved by $\cE$ is
that $\cC$ be isometric 
%(see Definition \ref{def:isometric})
for $\cE$.}\\
%\end{condition}

These two criteria form the foundation of our framework.  To
illustrate their application, here are some examples both simple and
subtle.

\begin{example} 
\label{ex:ClassicalFourState}
Suppose $\cS$ is a classical system with four states labeled
$\{0,1,2,3\}$, each perfectly distinguishable from the others.  $\cS$
passes through a channel that maps state $k$ randomly to $k$ or $k+1$
(mod 4), represented as a stochastic map
$$\cE = \begin{pmatrix} \hf&0&0&\hf \\ \hf&\hf&0&0 \\ 0&\hf&\hf&0 \\
0&0&\hf&\hf .
\end{pmatrix}$$ 
\noindent 
A stochastic map's information-preserving properties can conveniently
be represented by an \emph{adjacency graph} for the input states,
where state $j$ is connected to state $k$ if $\cE(j)$ overlaps with
$\cE(k)$.  This map's adjacency graph is:

\begin{center}
\includegraphics[height=25mm]{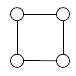}
\end{center}

The code $\cC_4 = \{0,1,2,3\}$ representing all information about
$\cS$ is not preserved, because $0$ and $1$ are perfectly
distinguishable, but $\cE(0)$ and $\cE(1)$ overlap.  A smaller code
$\cC_2 = \{0,2\}$ is preserved, even though neither $0$ nor $2$ is a
fixed point.  The code $\cC'_2 = \{1,3\}$ is also preserved, but the
union of $\cC_2$ and $\cC'_2$ is not preserved.  This demonstrates
that the set of preserved codes is \emph{not} convex; distinct
preserved codes may rely on mutually contradictory preconditions on
$\cS$, {\em e.g.}, ``$\cS$ was prepared in 0 or 2'' and ``$\cS$ was
prepared in 1 or 3.''
\end{example}

\begin{example}
\label{ex:WhyConvex}
Why must distinguishability be preserved, not just between code
states, but between convex combinations of them?

Let $\cE$ be a classical stochastic map on three states $\{0,1,2\}$,
which fixes states $0$ and $1$, but maps $2\to1$.  This map
``squashes'' the classical 3-simplex onto one of its sides, as in the
figure below.  Now, consider a code $\cC$ comprising the states on the
thick (red) line in the figure:
\begin{center}
\includegraphics[height=35mm]{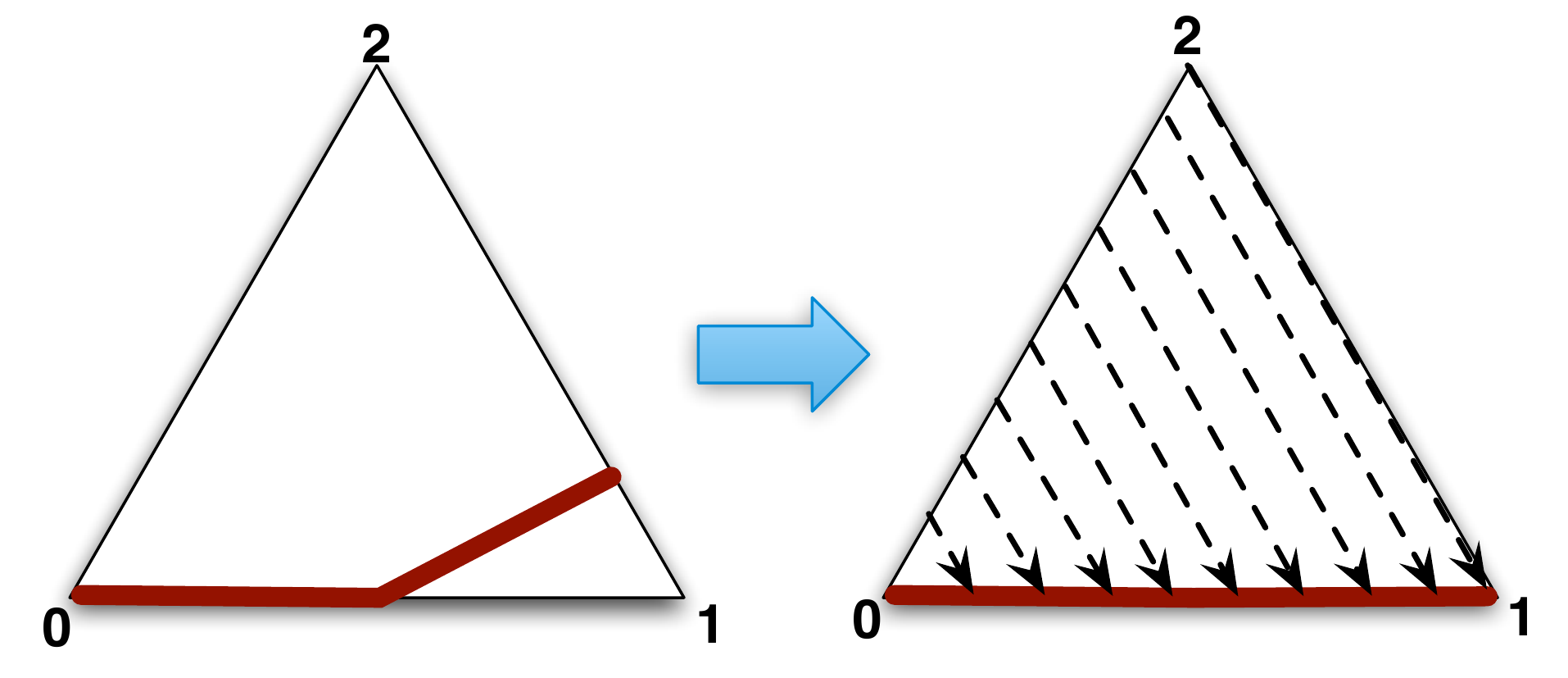}
\end{center}
This code is not preserved by $\cE$, because the original code has
structure that is missing in its image: States not on the line between
``0'' and ``1'' can be unambiguously discriminated (with $p>0$) from
states lying on the line.  However, there is no way to recover this
structure by applying another linear map afterward!  Still, if we
ignore convex combinations, then all the $1$-norm distances $\Vert
p\rho-(1-p)\sigma\Vert_1$, for $\rho,\sigma\in\cC$ are in fact
preserved by $\cE$.  This is because the best way to distinguish any
two states in $\cC$ is to measure $0$ vs. $\{1,2\}$, and because the
channel maps $2\to1$, it does not actually affect this measurement.
If we consider convex combinations, however, we see that $\cC$ is not
isometric to $\cE(\cC)$, resolving the problem.
\end{example}

\begin{example}
\label{ex:WhyWeights}
Why must all the \emph{weighted} 1-norm distances be preserved, rather
than just $\Vert \rho-\sigma\Vert_1$?

Let $\cH_3 = \bC^3$ be the state space of a qutrit.  Define $\cE$ to
be the channel that does nothing to the $\{\ket{0},\ket{1}\}$
subspace, but maps $\proj{2}\to\frac12(\proj{0}+\proj{1})$.  Now,
consider a code $\cC$ comprising all the states of the form
$$\rho = \frac12\left(\proj{\psi}_{\mathrm{Span}(\ket{0},\ket{1})} +
\proj{2}\right).$$ 
\noindent 
We can think of this code as the set of states that would be prepared
by a machine that is \emph{supposed} to produce qubit states in the
$\{\ket{0},\ket{1}\}$ subspace, but fails 50\% of the time and
produces $\proj{2}$ instead.

As in Example \ref{ex:WhyConvex}, this code is not preserved by $\cE$.
In this case, the problem is that Alice can check to see whether the
preparation failed or not, but Bob cannot.  As before, this intuition
is borne out by the fact that no recovery operation exists.  However,
if we compute the \emph{unweighted} 1-norm distances $\Vert
\rho-\sigma\Vert_1$, both before and after $\cE$ is applied, then we
find that they are unchanged.  Only when we require preservation of
the weighted 1-norm distances (corresponding to distinguishing states
with the aid of prior information), do we correctly derive that $\cC$
is not preserved.
\end{example}

As Example \ref{ex:WhyConvex} demonstrates, it is important that $\cE$
preserves distinguishability not just between states in $\cC$, but
between convex combinations of them.  This means that we can (without
loss of generality) \emph{extend $\cC$ to include all states in its
convex closure}.  From now on, we will simply assume that any
preserved code is convex in this sense, in line with \cite{BNPV}.  The
Weak Condition then has a simple geometric interpretation.  $\cE$ must
preserve the 1-norm distance between any two \emph{unnormalized}
states $p\rho$ and $(1-p)\sigma$.  This means that the entire convex
cone of $\cC$ -- that is, the set $\cC_+$ containing $x\rho$ for all
$x\geq0$ and $\rho\in\cC$ -- must be \emph{isometric} to its image
$\cE(\cC_+)$.  Two sets are isometric if there is a
distance-preserving mapping (an isometry) between them.  Here, the
relevant metric is the 1-norm distance
$$D(A,B) \equiv \Vert A-B \Vert_1,$$
\noindent 
and $\cE$ is the isometry that preserves it.  Thus, preservation
requires that the convex cone $\cC_+$ evolves {\em rigidly}, with
respect to the 1-norm distance, under $\cE$.

Our necessary and sufficient conditions bracket the as-yet-vague
notion of a code being preserved by a channel.  Fixedness seems too
strong, isometry perhaps too weak.  One of our main goals in this
paper is to derive a single, rigorously stated condition for
information to be ``preserved'' by a channel.  We will eventually do
so by squeezing the Strong and Weak Conditions together as follows:

\begin{proposition}
\label{prop:Preservedness}
If $\cC$ is a maximum isometric code for $\cE$ (i.e., it satisfies the
Weak Condition, and there is no larger $\cC$ that satisfies the Weak
Condition), then there exists a CPTP map $\cR$ such that
$\cR\circ\cE(\rho) = \rho$ for all states $\rho\in\cC$.
\end{proposition}

%In other words, we'll prove that whenever $\cC$ is isometric for $\cE$, there is a CPTP \emph{recovery} map $\cR$ that Bob could apply to his system, and which would map $\cE(\cC)$ back to $\cC$.  Bob can convert any isometric code for $\cE$ into a fixed code for $\cR\otimes\cE$.  Thus, isometry is actually a surprisingly strong condition, and if Bob's abilities are constrained only by physical law, then isometry clearly implies preservation.  The ability to apply such a recovery is a powerful resource, but a reasonable one.  Moreover, the question of whether $\cE$ itself preserves information should be independent of what the recipient can do.  Thus, we will conclude that Condition \ref{NecCond2} is necessary \emph{and} sufficient for information to be ``preserved''\footnote{One might protest and demand that the right definition of a preserved code should be based on the stronger Condition \ref{NecCond1} which arises from our considerations of the fundamental information retrieval procedure.  Here, again, we can only justify the use of the (much more precise and easy to use) alternative Condition \ref{NecCond2} retroactively, by proving Proposition \ref{prop:Preservedness}.  Hence, even as we use Definition \ref{def:pres}, the reader may wish to keep in mind that we will eventually show that it implies Condition \ref{NecCond1}.}.  We will reserve other names for other operational scenarios with more stringent criteria for preservation.

By proving this proposition, we will demonstrate that the strong and
weak conditions for preservation are equivalent -- \emph{given} the
ability to apply a recovery operation.  The proof is somewhat
involved.  In the next section, we will derive a structure theorem for
preserved codes, explore its consequences, and finally derive
Proposition \ref{prop:Preservedness} as as corollary (Corollary
\ref{cor:FPCond}) of Lemma \ref{lem:NSStruct}, which follows from
Theorem \ref{thm:PresCorr}.  Anticipating this sequence of
derivations, we proffer the following definition of ``preserved'' now,
with the understanding that it will only be justified by what follows:

\begin{definition}
\label{def:pres}
A code is \textbf{preserved} by a CPTP $\cE$ if and only if it
satisfies the Weak Condition -- that is, 
$$\Vert \cE(p\rho-(1-p)\sigma)\Vert_1=\Vert
p\rho-(1-p)\sigma\Vert_1,$$ 
\noindent 
for all $\rho,\sigma\in\cC$ and $p\in[0,1]$.
\end{definition}

%There is no circularity here, for up until now we have neither defined ``preserved'' nor based any other result upon its putative meaning.  It is merely not yet clear that an isometric code \emph{deserves} to be called ``preserved''.  While one property of the code (the ability to distinguish between states) is clearly preserved, we need to prove Proposition \ref{prop:Preservedness} before we can be certain that \emph{all} aspects of the code are preserved.  So, let us now begin to justify the terminology by exploring the structural consequences of Definition \ref{def:pres}.

%**********************************************************************
%**********************************************************************
\section{The structure of preserved information}
\label{sec:Struct}

In Section \ref{sec:PresInfo}, we stated plausible necessary and
sufficient conditions for a code to be ``preserved'', and suggested a
formal definition of preservation (conditional on some technical
results to be proved in what follows).  Next, we shall build upon this
foundation, elucidating the structures that follow from it.  First, we
will prove a series of theorems about preserved codes, culminating in
a structure theorem showing that preserved codes have the same
``shape'' as matrix algebras.  This indicates that preserved codes are
related to algebras, but provides no real context for \emph{how} they
are related, nor what role the algebra is playing.  So, our second
task is to analyze the underlying IPS.

Except where explicitly noted, all the proofs of theorems and lemmas
in this section have been deferred to Appendix \ref{sec:Proofs}.

%*********************************************************************
\subsection{The shape of a preserved code}

Suppose that 
%$\cE$ is a channel that acts on states from a Hilbert space $\cH$, and 
$\cC$ is a preserved code for $\cE$.  Starting from Definition
\ref{def:pres}, what can we derive about $\cC$?  Quite a lot, as it
turns out.  The following two definitions from Ref. \onlinecite{BNPV}
will be needed.

\begin{definition}
\label{def:noiseless}
A code $\cC$ is \textbf{noiseless} for a CPTP $\cE$ if and only if it
is preserved by any convex combination $\sum_nq_n\cE^n$, with $q_n\geq
0$ and $\sum_nq_n=1$.
\end{definition}

Noiselessness is stricter than preservation (every noiseless code is
preserved, but most preserved codes are not noiseless), but weaker
than fixedness (every fixed code is noiseless, but some noiseless
codes are not fixed).  Noiseless codes are special because their
states remain distinguishable no matter how many times $\cE$ is
applied (note that only channels whose output space is the same as
their input space can have noiseless codes).  This captures the
operational significance of fixedness -- and as we will show below
(Lemma \ref{lem:NSFixed}), there is a close mathematical connection
between noiseless and fixed codes.

\begin{definition}
\label{def:correctable}
A code $\cC$ is \textbf{correctable} for $\cE$ if and only if there
exists a CPTP $\cR$ such that $\cC$ is noiseless for $\cR\circ\cE$.
\end{definition}

Correctable codes can be \emph{made} noiseless, by applying a suitable
correction operation every time $\cE$ happens.  Readers familiar with
QEC may worry that our definition is slightly different from the usual
one, which requires that $\cC$ be \emph{fixed} by $\cR\circ\cE$,
rather than just noiseless.  It will turn out that our (apparently
weaker) condition implies the usual one, so we obtain the same result
with a weaker assumption\footnote{In the terminology of
Ref. \onlinecite{Ticozzi2010}, a code $\cC$ which is fixed by
$\cR\circ\cE$ is referred to as ``completely correctable''.  That
complete correctability is in fact equivalent to correctability can be
alternatively established by exploiting the explicit form of
$1$-isometric encodings, see Thm. 4 therein.}.  We are now in a
position to state a key theorem:

\begin{numtheorem}
{\ref{thm:PresCorr}}
\THMPRESCORR
\end{numtheorem}

Although the full proof is rather technical (see Appendix
\ref{sec:Proofs}), one aspect is especially useful and interesting.
We prove the theorem by \emph{explicitly} constructing a correction
operation for an arbitrary code $\cC$.  Moreover, the correction
operation is independent of $\cC$'s structure, and depends only on
$\cC$'s support. A code's support is the subspace $\cP\subseteq\cH$,
comprising the union of the supports of all $\rho\in\cC$.  Since the
correction only depends on the code's support, every code with the
same support will be corrected by the same operation. Remarkably, this
operation coincides with the \emph{transpose channel} introduced in
Ref. \onlinecite{Barnum00}, defined as
\begin{equation}
\cEP=\Pi\circ\cE^\dagger\circ\cN,
\label{transpose}
\end{equation} 
where $P$ is the projector onto $\cP$, $\Pi(\cdot)=P\cdot P$ is the
projection onto $\cP$, $\cE^\dagger$ is the adjoint map of $\cE$, and
$\cN$ is a normalization map
$\cN(\cdot)=\cE(P)^{-1/2}(\cdot)\cE(P)^{-1/2}$.
  
This theorem has two consequences.  First, it strongly suggests that
Definition \ref{def:pres} captures the critical notions of information
preservation.  Second, it implies a simple corollary: Every preserved
code for $\cE$ is noiseless for some other map $\cR\circ\cE$.  This
connection from preserved to noiseless codes is a step toward proving
Proposition \ref{prop:Preservedness}.  Even more importantly, it will
let us derive a structure theorem for preserved codes.  To do so, we
need another result.

\begin{numlemma}{\ref{lem:NSFixed}}
\LEMNSFIXED
\end{numlemma}

This means that noiseless and fixed codes are geometrically 
almost the same.  A noiseless code does not have to be precisely
fixed, but it will always be isometric to a fixed code -- that is, it
will have the same shape.  A simple example may be in order.

\begin{example}
\label{ex:NSnotFixed}
Let $\cE$ be a channel on two qubits, labeled $A$ and $B$, that does
nothing to $A$ but depolarizes $B$:
$$\cE(\rho_{AB}) = \Tr_B(\rho_{AB})\otimes\frac{{\Id}_B}{2}.$$
Qubit $A$ clearly is a NS under $\cE$, whose fixed states are of the
form $\cC_{NS}=\rho_A\otimes\left(\frac{\Id}{2}\right)_B$. However,
there are other noiseless codes.  For instance, let $\cC$ comprise all
states of the form $\rho_A\otimes\proj{0}_B$.  Qubit $B$ carries no
information, so $\cE$'s action on it is irrelevant.  None of $\cC$'s
distinguishability properties are affected by $\cE$, even though $\cC$
is not actually fixed.  Note, however, that $\cC$'s image $\cE(\cC)$
\emph{is} a fixed code.  Repeated applications of $\cE$ map its
noiseless codes to fixed codes.
\end{example}

Lemma \ref{lem:NSFixed} implies that a channel has a unique maximum (largest)
noiseless code, and that the latter must be isometric to the set of
\emph{all} fixed states:

\begin{numcorollary}{\ref{cor:NSMax}}
\CORNSMAX
\end{numcorollary}

A channel can have smaller noiseless codes -- even maximal ones.
Consider the following example:

\begin{example}
Let $\cE$ be a channel on two qubits, labeled $A$ and $B$, acting as
follows: It measures $B$ in the $\{\ket{0},\ket{1}\}$ basis;
conditional on $\proj{0}$ it does nothing; conditional on $\proj{1}$,
it dephases $A$ and flips $B$ to the $\ket{0}$ state.  Every state of
the form $\rho_A\otimes\proj{0}_B$ is a fixed point, and so the
largest noiseless code encodes a single qubit in $A$, like in Example
\ref{ex:NSnotFixed}.  However, there is another maximal noiseless code
comprising all states of the form
$(p\proj{0}+(1-p)\proj{1})_A\otimes\proj{1}_B$.  It is isometric to a
strict subset of the fixed points, so it is not a maximum code.
\end{example}

Recall that any preserved code can be made noiseless, by applying a
suitable recovery map (Thm. \ref{thm:PresCorr}).  By combining this
theorem with the corollary to Lemma \ref{lem:NSFixed}, we establish a
direct connection between arbitrary preserved codes and fixed states
of CPTP maps.

\begin{numtheorem}{\ref{thm:PresFixed}}
\THMPRESFIXED
\end{numtheorem}
\begin{proof}
This follows from combining Lemma \ref{lem:NSFixed} with Theorem
\ref{thm:PresCorr} and Definition \ref{def:correctable}.
\end{proof}

This points the way to the structure theorem we are looking for,
provided that we can say something about the fixed points of the
unknown CPTP map $\cR\circ\cE$. Quite a bit is known about fixed
points of CPTP maps.  In particular, if $\cH$ is finite-dimensional,
and the map is \emph{unital} (meaning that it preserves the identity
operator), then its fixed points form a matrix algebra
\cite{AGG02a,KribsPEMS03}.

A matrix algebra (a.k.a. finite-dimensional $C^*$-algebra) is a vector
space of complex matrices, closed under multiplication and Hermitian
conjugation.  It follows that
\begin{enumerate}
\item The matrices must be square (otherwise they cannot be multiplied);
\item The set of \emph{all} $d\times d$ complex matrices ({\em i.e.},
operators on a $d$-dimensional Hilbert space $\cH$) is an algebra,
denoted $\cM_d$ or $\cM_{\cH}$ henceforth;
\item The set containing only the $d\times d$ identity matrix is an
algebra, denoted ${\Id}_d$ or ${\Id}_{\cH}$.
\end{enumerate}
Happily, these three simple facts are sufficient to describe
\emph{any} matrix algebra.  The structure theorem \cite{Dav96a} for
matrix algebras states that any such matrix algebra $\cA$ is unitarily
equivalent to the canonical form:
\begin{equation}
\label{decomp}
\cA\simeq\bigoplus_k\cM_{A_k}\otimes{\Id}_{B_k},
\end{equation}
where $A_k$ and $B_k$ are complex vector spaces of dimension $d_k$ and
$n_k$, respectively.  We will refer to each of the subspaces
$A_k\otimes B_k$ in the direct sum labeled by $k$ as a ``$k$-sector''.
Each $k$-sector factors into a \emph{noiseless subsystem} (with
Hilbert space $A_k$) and a \emph{noise-full subsystem} (with Hilbert
space $B_k$)\footnote{Note that in the original definition of
\cite{KnillPRL00}, a decomposition of the form given in
Eq. (\ref{decomp}) is applied to the (associative) {\em error algebra}
as opposed to states, whereby the identification of the noiseless
factors with $B_k$.}  Thus, every matrix algebra is built up out of
the two simple components described in points 2. and 3. above (the
algebra of all $d\times d$ matrices, and the trivial algebra).

As remarked earlier, the fixed points of a unital map form an algebra.
Prior to this work (and the results anticipated in \cite{BNPV}), no
such result was known for\emph{arbitrary} non-unital maps.
%So we proved one.  
Before stating our main structure theorem, we need to define a couple
of terms.

\begin{definition}
\label{def:Distortion}
Consider a matrix algebra $\cA = \bigoplus_k\cM_{A_k}\otimes{\Id}_{B_k}$,
which induces a Hilbert space decomposition $\cH =
\bigoplus_k{A_k\otimes B_k}$.  A \textbf{distortion map} for $\cA$ is
a CPTP map $\cD$ such that, for every $X =
\sum_k{M_{A_k}\otimes{\Id}_{B_k}}$ in $\cA$,
$$\cD(X) = \sum_k{M_{A_k}\otimes\tau_{k}},$$ 
\noindent 
where $\tau_{k}$ is a positive semidefinite matrix on $B_k$ that does
not depend on $M_{A_k}$.  $\cD(\cA)$ is a \textbf{distortion} of
$\cA$.  A vector space of matrices $\tilde{A}$ is a \textbf{distorted
algebra} if it is a distortion of some matrix algebra $\cA$.
\end{definition}
A distorted algebra is simply an algebra in which each identity factor
has been replaced with an arbitrary (but fixed) matrix $\tau_k$.  A
distorted algebra is not an algebra under standard matrix
multiplication (because $\tau_k^2 \neq \tau_k$), although it is under
a suitably redefined matrix multiplication.  More importantly, there
exist CP distortion maps that reversibly transform
$\tilde{\cA}\leftrightarrow\cA$, simply by changing the $\tau_k$
factors.  Thus, $\tilde{\cA}$ and $\cA$ are isometric.

We can now characterize the fixed points of an arbitrary CPTP map
\emph{and} its adjoint (that is, fixed states \emph{and} observables): 

\begin{numtheorem}{\ref{thm:FixedPtThm}}
\THMFIXEDPTTHM
\end{numtheorem}

While Theorem \ref{thm:FixedPtThm}
is somewhat intimidating (we shall use all of its pieces in Section
\ref{sec:App}), the payoff for its complexity is that it consistently
unifies the Schr{\"o}dinger and Heisenberg pictures of information
preservation (see also Refs. \onlinecite{KnillPRL00,NSqubit,Beny}).  The
Schr{\"o}dinger approach involves looking at the fixed states in
$\FixE$.  The Heisenberg approach, on the other hand, emphasizes
\emph{observables} of the system, which evolve according to
$\cE^\dagger$ (since expectation values evolve as
$\Tr\{X\cE(\rho)\}=\Tr\{\cE^\dagger(X)\rho\}$). Fixed states of $\cE$
in the Schr{\"o}dinger picture translate to fixed observables of
$\cE^\dagger$ in the Heisenberg picture.  Theorem \ref{thm:FixedPtThm} shows that
{\em both} such fixed sets are isometric to the \emph{same} matrix
algebra $\cA$.  This algebra determines the structure of preserved
codes, so the two pictures (interpreted correctly) yield equivalent
characterizations of preserved information.

Some of the results in Theorem \ref{thm:FixedPtThm} were proved
previously, in different (though related) contexts.  Our
characterization of $\FixEdag$ [parts (ii) and (iv)] follows, in
particular, from a classic operator algebra paper by Choi and Effros
\cite{Choi77}.  Their results are substantially more abstract and less
constructive, but Kuperberg subsequently applied them to quantum
information (see Ref. \onlinecite{KuperbergIEEE03}, Theorems 2.2 and 2.3).
The proofs given here are self-contained (and perhaps more accessible
to physicists). 

The fact that an arbitrary CPTP map's fixed points are isometric to a
matrix algebra, together with Theorem \ref{thm:PresFixed}, nails down the
structure of \emph{every} preserved code.  If $\cC$ is a preserved
code for a channel $\cE$, then it is isometric ({\em i.e.}, rigidly
equivalent) to a matrix algebra.  Furthermore, $\cE$'s fixed points are a
subspace of matrices that looks very much like an algebra --
\emph{except} that each of the identity factors
${\Id}_{B_k}$ has been replaced by some fixed matrix $\tau_k$.

While the domain of $\cE$ contains all operators on $\cH$, its
physical significance comes from its action on positive semidefinite
states.  Given any algebra $\cA$ in the canonical form of
Eq. (\ref{decomp}), we can easily identify the set $\cA_+$ of
positive states in $\cA$: $\cA_+$ contains states of the form
$\sum_k{p_k\rho_k\otimes\left(\frac{{\Id}_{B_k}}{n_k}\right)}$,
where the $\{p_k\}$ form a probability distribution, and the
$\{\rho_k\}$ are arbitrary states on the noiseless factors.

$\cE$'s fixed states ($\FixE_+$) form a very similar set, comprising
states of the form $\sum_k{p_k\rho_k\otimes\tau_k}$, where the $\{p_k\}$
and $\{\rho_k\}$ are probabilities and arbitrary states as above, and the
$\tau_k$ are {\em fixed} density matrices determined by $\cE$.  Any set of
fixed states is a fixed code for $\cE$, and $\FixE_+$ is the unique
largest fixed code.  Lemma \ref{lem:NSFixed} implies a relationship between
noiseless and fixed codes, from which it follows that:

\begin{numlemma}{\ref{lem:NSStruct}}
\LEMNSSTRUCT
\end{numlemma}

Note that the lemma is only proved for channels with a \emph{full-rank
fixed point}.  We believe that a similar result can be proved for
arbitrary channels, but there are some tricky details that obscure the
main point.  We only need to apply this result to channels of the form
$\RP\circ\cE$, with $\RP$ defined in Eq. (\ref{transpose}).  Each such
channel, from $\cB(\cP)\to\cB(\cP)$, is actually unital (since $\RP
\circ \cE (P)=P$), so it has a full-rank fixed point, and Lemma
\ref{lem:NSStruct} is sufficient to characterize its noiseless codes:
They are isometric to the channel's fixed points, and those have
algebraic structure.

So while a channel $\cE$ typically has a lot of noiseless codes,
they turn out to be trivial variations on a constant theme.  The
variation is a \emph{gauge} -- a particular state $\mu_{B_k}$ for each
of the noise-full subsystems.  The actual information is carried by the
variation in the code states, which differ only on the noiseless factors
$A_k$, and in the weights $p_k$ assigned to the different $k$-sectors.
This suggests an obvious way to turn noiseless codes into fixed codes,
simply by adjusting the state of the noise-full subsystems.  Thus, we can
finally justify Proposition \ref{prop:Preservedness} with the following
corollary to Lemma \ref{lem:NSStruct}.

\begin{numcorollary}{\ref{cor:FPCond}}
For every maximum preserved code $\cC$, there exists a CPTP map $\cR$
such that $\cR\circ\cE(\rho)=\rho$ for all states $\rho\in\cC$.
\end{numcorollary}

We have finally proved the central proposition of the previous
section, justifying our definition of ``preserved''.  If and only if a
code satisfies Definition \ref{def:pres}, there exists a recovery
operation that makes it into a fixed code, which is clearly preserved
in the strongest possible sense.  However, this depends on Bob's
ability to apply the necessary recovery immediately after $\cE$
happens!  Section \ref{sec:Types} considers the effect of placing
operational restrictions on what Bob can do, and how this can change
the criteria for preservation.

We note in passing that the framework presented by Kuperberg in
\cite{KuperbergIEEE03} is similar and uses much of the same
mathematics.  However, it only addressed noiseless and unitarily
noiseless information (a.k.a. \emph{infinite-distance} codes), not
correctable information, or the relationship between preservation
and correctability.

%******************************************************************
\subsection{IPSs: The structures that underly preserved codes}

%In this section, we've presented most of our technical results and
%theorems (the exceptions have to do with operational criteria that we
%haven't discussed yet).  
Taken together, the results we have presented thus far indicate a {\em
rigid algebraic structure} lurking within each CPTP map $\cE$, which
constrains the shape of its preserved and noiseless codes.  The codes
themselves are not the structure, however.  There are many noiseless
codes, all distortions of the same algebra.  What matters is
their shared structure.  In fact, all these noiseless codes are
manifestations of a unique noiseless IPS underlying the channel, which
we turn to explore next.  We begin with an example. 
%Furthermore, a channel may have many
%preserved codes, each isometric to some algebra.  
%Suppose that we have a noisy channel $\cE$ acting on states from a
%Hilbert space $\cH$, and it has some noiseless codes.  These codes
%share some features -- in fact, they all have the same shape.  

\begin{example}
\label{ex:ManyCodesOneIPS}
Consider the two-qubit channel of Example \ref{ex:NSnotFixed}, which
depolarizes qubit $B$.  There is an infinite family of maximum
noiseless codes for this channel: If $\tau_B$ is a valid state for
$B$, then $\cC_\tau \equiv \{\rho_A\otimes\tau_B\ \forall\ \rho_A\}$
is a noiseless code.  While distinct, these noiseless codes are all
equivalent, and share the same recovery operation, $\cR ={\Id}$.
Thus, they are all manifestations of the same noiseless IPS.
\end{example}

This example demonstrates a noiseless IPS, but a channel can also have
correctable codes that are not noiseless.  However, these codes are
noiseless for the appropriate $\cR\circ\cE$, so the preserved codes
with a common recovery $\cR$ also share a common structure.  A channel
can have multiple preserved IPSs.  In a way, each IPS is akin to a
hole in the wall of noise, through which information can (if properly
aimed) pass unscathed.  The preserved codes reflect this structure,
but their diversity can also obscure it.  If we can concisely describe
a channel's IPSs, we have (for all practical purposes) completely
classified its preserved codes.

Let us define ``information-preserving structure'' more precisely.
Every maximum preserved code is isometric to an algebra, and preserved
codes isometric to the same algebra are essentially trivial variations
on a theme.  They are manifestations of the same underlying IPS.

%\todo{Note from RBK:  I've added a semantic here, which is that all the codes associated with a single IPS must have the same recovery.  Please think about this and see if it makes sense.  See following example, too.}

\begin{definition}
An \textbf{information-preserving structure} for a CPTP map $\cE$ is
an equivalence class of maximum preserved codes for $\cE$.  Two codes
are equivalent if they are isometric to the same algebra, and are
preserved according to the same operational criterion ({\em e.g.},
Definition \ref{def:pres}, Definition \ref{def:noiseless}, or one of
the other operational criteria in Section \ref{sec:Types}) with the
same recovery operation.
\end{definition}

The IPS is {\em not} itself an algebra.  Rather, an IPS is an abstract
structure (an equivalence class of codes), whose properties are
defined by an associated algebra.  It is possible for a channel to
have two distinct IPS with the same (isomorphic) algebra.

By looking at the structure theorem for matrix algebras
(Eq. \ref{decomp}), we can interpret any given IPS.  It consists of
one or more $k$-sectors, each of which contains a noiseless subsystem
supported on $A_k$ and a noise-full subsystem supported on $B_k$. Any
information encoded into the $A_k$ factors will be preserved by $\cE$,
whereas any information encoded into the $B_k$ factors is irreparably
damaged.  The information-carrying capability of a code is determined
entirely by its underlying IPS; distinct codes that share an IPS are
equivalent, carrying the same kind and amount of information.

\begin{example}
\label{ex:MultipleIPS}
Consider a classical stochastic map on four symbols, $\{0,1,2,3\}$,
which maps each input symbol to a mixture of output symbols as follows
$$ 0\to\{0,1\},\:\; 1\to\{2,3\}, \:\;2\to\{0,2\},\:\; 3\to\{1,3\}.$$
%\begin{itemize}
%\item $0\to\{0,1\}$
%\item $1\to\{2,3\}$
%\item $2\to\{0,2\}$
%\item $3\to\{1,3\}$
%\end{itemize}
\noindent 
There are exactly two maximal preserved codes for this channel, both
of which are actually noiseless: $\{0,1\}$ and $\{2,3\}$.  They are
equivalent, and both described by the same (commutative) algebra -- but
this is merely a coincidence.  The two codes occupy disjoint subspaces
of the input, they both get mapped to output states which span the
entire output space in different ways, they have entirely different
recovery maps, and by changing the channel slightly, we can easily
eliminate either code without affecting the other.  They are thus not
manifestations of the same IPS.
\end{example}

To make use of an IPS, Alice and Bob use any of the equivalent codes
associated with that IPS.  Each of these codes is isometric to the
IPS's algebra, so the structure of that algebra tells us everything
about its information-carrying capability.  Since the algebra can be
decomposed according to Eq. \eqref{decomp},
\begin{equation*}
\cA\simeq\bigoplus_k\cM_{A_k}\otimes{\Id}_{B_k},
\end{equation*}
we can represent it concisely by its \emph{shape}: the vector
$\{d_1,d_2,\ldots,d_n\}$ listing the dimensions of the
information-carrying factors $\cH_{A_k}$ (the noise-full factors are
irrelevant). Pictorially: 

\begin{center}
\includegraphics[width=2.5in]{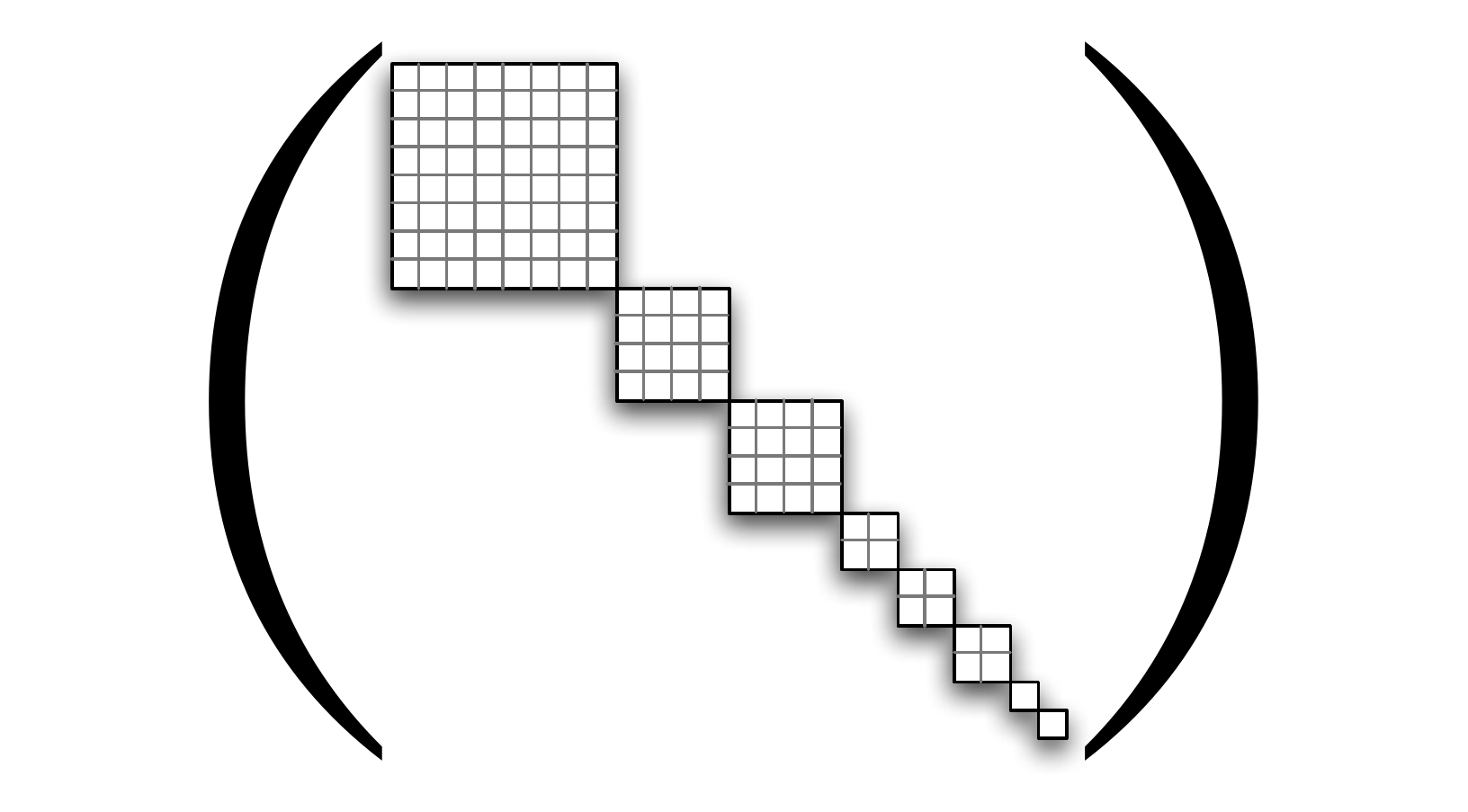}
\end{center}

The IPS shape characterizes the type and amount of information an IPS
can carry.  A $k$-sector with a $\cH_{A_k}$ factor of dimension
$d_k>1$ can carry quantum information.  Classical information is
carried by the choice between the different $k$-sectors.  Kuperberg,
in Ref. \onlinecite{KuperbergIEEE03}, described such a noiseless IPS as a
\emph{hybrid quantum memory}, capable of simultaneously storing or
transmitting a certain amount of quantum information and a certain
amount of classical information.  The IPS shape provides a very
concise way of describing the noise-free degrees of freedom within a
given system's Hilbert space -- much more convenient than listing the
$d^4$ real parameters required to specify a quantum process on a
$d$-dimensional Hilbert space!

From a physical standpoint, algebraic structure imposes a very strong
constraint on the types of information that a quantum process can
preserve.  {\em A priori}, we might suppose that any subspace of
$\cB(\cH)$ could be ``superselected'' by some process, however the
theorems proved above rule out most such possibilities.

\begin{example}
\label{ex:NoPancake}
Consider a single qubit, with $\cH={\mathbb C}^2$. Its dynamics will
be described by some CPTP map (or family of them).  These dynamics
destroy some information while preserving other information,
a.k.a. \emph{dynamical superselection}.  Although there are infinitely
many different kinds of dynamics, there are only three possible IPSs.
The dynamics can preserve the full qubit algebra $\cM_2$; or a
classical bit, represented (up to unitaries) by the algebra
$span\{{\Id},\sigma_z\}$; or nothing, represented by the trivial
algebra $\{ {\Id}\}$.  In particular, there are \emph{no}
CP maps that single out a \emph{rebit} (a mythical physical system
described by a 2-dimensional real Hilbert space).  This would
correspond to preserving information on some equatorial plane of the
Bloch sphere, spanned by $\sigma_x$ and $\sigma_y$, while annihilating
information about $\sigma_z$.  But $span\{\sigma_x,\sigma_y\}$ is not
a closed algebra, for $\sigma_x$ and $\sigma_y$ generate the full
qubit algebra.  The fact that no CPTP map can annihilate $\sigma_z$
while preserving $\sigma_x$ and $\sigma_y$ is known, in quantum
information folklore, as the ``No-Pancake Theorem''.
\begin{center}
\includegraphics[height=35mm]{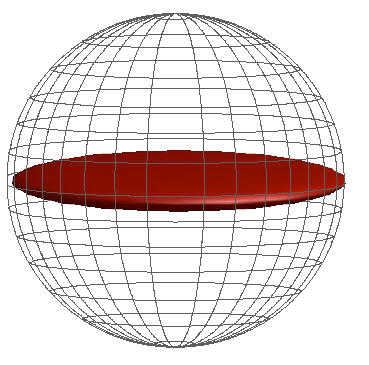}
\end{center}
Our central result might be thought of as a fully general No-Pancake
Theorem, since it rules out the dynamical superselection of all such
non-algebraic IPS.
\end{example}

We can safely talk about ``qu\emph{d}its'' of information within the
code, specified by the IPS shape.  Each qudit corresponds to a logical
subsystem -- a $d$-dimensional Hilbert space within the full Hilbert
space, which need not correspond to a physical subsystem but is
nonetheless an independent quantum degree of freedom.  Multiple qudits
in a direct sum represent a \emph{classical} degree of freedom, for
while the different terms in the direct sum correspond to perfectly
distinguishable states, superpositions across them are not preserved.
We can use these rules to exhaustively catalog all the possible
degrees of freedom (up to unitary rotations) within any given quantum
system.

\subsection{Different kinds of IPS}
\label{sub:kinds}

We identified the Weak Condition as the weakest reasonable condition
for information to be preserved. It ensures that Bob can \emph{in
principle} restore the system's initial state -- but, if Bob has
limited resources, then he may be unable to do so \emph{in practice}.
Still, Bob's resources may be sufficient to correct a code that
satisfies some stronger condition.  Each operational constraint on Bob
defines some condition on $\cC$ that is necessary and sufficient for
it to be ``preserved'' in this situation.

One important example has already appeared, \emph{noiseless}
information (Definition \ref{def:noiseless}).  Noiseless codes require
no correction at all, so noiselessness is a very strong condition.  In
Section \ref{sec:Types}, we will consider several other conditions.
Each such condition defines a distinct class of IPSs.  So amongst one
or more preserved IPSs a channel may support, one may also be
noiseless.  A channel's noiseless IPS is unique, because of its
relationship to the channel's fixed points (see also Section
\ref{sec:App} for further discussion of this point).

Most of the commonly studied techniques for information preservation
correspond either to a noiseless IPS, or to a preserved/correctable
IPS.  Three of the ``canonical'' structures that we mentioned in
Section \ref{sec:PresInfo} -- pointer bases, DFSs, and NSs --
correspond to noiseless IPS.  Pointer bases have the shape
$\{1,1,1\ldots\}$, describing a complete set of 1-dimensional
$k$-sectors (both $A_k$ and $B_k$ are trivial for all $k$).  A DFS has
the shape $\{d\}$, describing a single $k$-sector with a trivial
$\cH_{B_k}$.  A NS has the same shape $\{d\}$, but it corresponds to
the $A_k$ factor of a single $k$-sector with a nontrivial co-factor
${B_k}$.

The relationship between a NS defined in the traditional way as
discussed in Example \ref{ex:NS} and a noiseless IPS as defined in
\cite{BNPV} and in this paper, has some subtleties.  A noiseless IPS
rests upon a family of noiseless codes, or sets of states, whereas the
traditional definition of a NS makes no direct reference to sets of
states.  The correspondence between the two frameworks arises because
Eq. (\ref{defNS2}) is satisfied if and only if there exist noiseless
codes.  This does \emph{not} imply that Eq. (\ref{defNS2}) has
anything directly to do with noiseless codes!  In particular, a set of
states $\{\rho_{AB}\}$ satisfying Eq. (\ref{defNS2}) need {\em not} be
a noiseless code.

\begin{example}
Consider a bipartite system $AB$ with Hilbert space $\cH_{AB} =
\cH\otimes\cH$, a channel $\cE$ that depolarizes system $B$ but leaves
$A$ untouched, and the set of states given by $\cC =
\{\ket\psi\otimes\ket\psi\}$ for all $\ket{\psi}\in\cH$.  Since
$$\cE(\proj{\psi}\otimes\proj{\psi}) =
\proj{\psi}\otimes\frac{{\Id}}{\mathrm{dim}(\cH)},$$ 
\noindent 
$\cC$ satisfies Eq. (\ref{defNS2}).  (In fact, every state $\rho_{AB}$
satisfies Eq. (\ref{defNS2}).)  Nonetheless, $\cC$ is not noiseless.
Eq. (\ref{defNS2}) merely guarantees that a noiseless code will exist.
\end{example}

Error-correcting codes are built upon preserved IPSs.  Most QECCs are
subspace codes, so a code with a recovery operation $\cR$ is a DFS of
$\cR\circ\cE$.  While {\em every} subspace code is associated to an NS
of $\cR\circ\cE$ (as implied by Thm. 6 in \cite{KnillPRL00}), an
operator code (OQECC) is also an NS of $\cR\circ\cE$, for the same
$\cR$.  In each case, the code is built upon the noiseless IPS of
$\cR\circ\cE$, \emph{not} of $\cE$ itself.  In fact, $\cE$ may
have no noiseless IPS at all.  However, since these codes are
correctable for $\cE$, they are preserved by it, and so they are
associated with preserved IPSs of $\cE$.

\begin{example}
Consider a system of 5 qubits, and a channel $\cE$ that picks one
qubit at random and depolarizes it.  This is precisely the error model
for which the 5-qubit QECC was developed
\cite{BennettPRA96,LaflammePRL96}, so $\cE$ has a 1-qubit preserved
IPS.  However, it has no noiseless codes at all, because repeatedly
applying $\cE$ will eventually depolarize all five qubits with high
probability.
\end{example}

Example \ref{ex:MultipleIPS} demonstrates that a channel can have more
than one preserved IPS.  Each is a noiseless IPS for some
$\cR\circ\cE$ (a consequence of Theorem \ref{thm:PresCorr}), and may
be associated with many preserved codes, all of which are corrected by
the same $\cR$.  We would like to have a procedure for listing, or at
least counting, all the IPS for a given channel -- but unfortunately
we do not know how to do this.

What we \emph{can} say (from Theorem \ref{thm:PresCorr}) is that
$\cE$'s IPSs comprise all the noiseless IPSs of $\cR\circ\cE$ for
\emph{all} CPTP maps $\cR$.  A simpler and stronger characterization
follows from the structure of the proof.  The correction operation for
a code depends only on the code's support, so every code with the same
support will be corrected by the same operation.  This yields a
simpler description: $\cE$'s IPSs comprise all the noiseless IPSs of
$\cEP\circ\cE$ for all subspaces $\cP\subseteq\cH$.

While this suggests a way of searching for IPSs (just try every
subspace, one at a time), there are uncountably many subspaces to
search (see \cite{ChoiPRL06}).  It may be possible to reduce this
problem to searching a countable, even finite set.  Unfortunately, it
is \emph{not} possible to do so efficiently.  Just finding the largest
classical code for an arbitrary channel is NP-hard, so listing all its
preserved IPS is at least this hard.  More precisely, let the size of
an IPS be measured by the total number of perfectly distinguishable
states in one of its preserved codes. Then we have the following:

\begin{numlemma}{\ref{lem:FindingIPSisHard}}
\LEMFINDINGIPSISHARD
\end{numlemma}

%******************************************************************
%******************************************************************
\section{Operational constraints and preserved codes}
\label{sec:Types}

Our focus thus far has been on a single notion of preservation.  We
assumed that Alice and Bob were unlimited in their actions (within the
laws of physics), and ended up with a preservation condition that
depended only on whether $\cE$ actually destroyed some of the
information.  In this section, we will relax this focus, and consider
the effect of restrictions on the sender and receiver.  Bob may not
want to correct the channel constantly, or he may not know how many
times $\cE$ has been applied.  Alice may have a faulty encoder -- or
perhaps she is not even cooperative.  Operational constraints of this
sort lead to alternative conditions for preservation.  We shall
discuss some of the most useful and interesting operational
constraints, and the corresponding types of IPS.

%The most important generalization concerns repeated application of the noise map $\cE$.  When $\cE$ is applied just once, the paradigm is simple:  we encode the information into the system, pass it through the noise channel, then try to retrieve the information.  We can extend this paradigm in two ways.  If $\cE$ maps a system to itself, then we can consider multiple applications of $\cE$ in between encoding and decoding.  Here, $\cE$ describes the (Markovian) noise afflicting our system in each time step $\Delta t$, and we can ask what information is preserved for a time $T = n\Delta t$.  The second variation applies even if $\cE$ has different input and output systems.  Here, we allow a recovery operation $\cR$ that restores information to a system of the same kind as the input (and on which $\cE$ can therefore act again), but demand that $\cR$ be chosen from a restricted set (e.g., only unitary operations, or only projective measurements).

%*******************************************************************
\subsection{Infinite-distance IPSs}

Suppose we want to store information in a physical system for a time
$T>0$, during which $\cE$ will be applied $n$ times.  Further, we
cannot perform any active operations on the system during this period.
Then the information carried by a code $\cC$ remains intact only if
$\cC$ is preserved by the channel $\cE^n$. If $T$ (or $n$) is unknown
in advance, $\cC$ has to be preserved by all possible powers of $\cE$.
One example of a channel for which this holds is a unitary channel:
\begin{equation}
\cE(\cdot)=U(\cdot)U^\dagger,
\end{equation}
for some unitary $U$.  A unitary channel adds no noise at all; it just
rotates the code around, and the actual rotation depends on how many
times it is applied.  As long as we know how many times $U$ has been
applied, we can recover \emph{any} initial state by applying $U^{-n}$.

This kind of behavior can be found even in channels that are not
purely unitary:

\begin{example}
\label{ex:UNS}
Consider a channel on two qubits, labeled $A$ and $B$, which applies a
unitary $U$ to qubit $A$ and depolarizes qubit $B$.  The channel is
not unitary, for it adds entropy to any pure state -- but nonetheless,
it acts unitarily on qubit $A$.  The code $\cC =
\left\{\rho_A\otimes\left(\frac{{\Id}_B}{2}\right)\ \forall\
\rho_A\right\}$ is preserved by any number of applications of $\cE$.
\end{example}

We shall refer to a code that remains preserved no matter how many
times $\cE$ is applied as \emph{unitarily noiseless} under
$\cE$. Formally, we define a unitarily noiseless code as in \onlinecite{BNPV}:

\begin{definition}
A code $\cC$ is \textbf{unitarily noiseless} under a CPTP $\cE$ if and
only if it is preserved by $\cE^n$ for any $n\in\mathbb{N}$.
\end{definition}

\noindent 
Notice that to retrieve the information stored in a unitarily
noiseless code, we need to know the value of $n$ or, equivalently the
length of time $T$, in order to construct the appropriate Helstrom
measurement.  In the previous example, if we lose track of $n$, then
qubit $A$ will get dephased in the diagonal basis of $U$.  Ensuring
that unitarily noiseless codes are preserved indefinitely requires a
good clock.

Are there codes for which we do not even need a clock? Certainly --
for instance, a code containing fixed states of $\cE$.  Such a code is
fixed not only by $\cE$, but also by $\cE^n$ for any $n$, and by any
convex combination $\sum_n{q_n\cE^n}$ (where $\{q_n\}$ is a probability
distribution).  So someone ignorant of $n$ can describe the process
by a mixture of different $\cE^n$, and information in a fixed code is
still preserved!  Moreover, only the information-carrying part of the
code needs to be invariant under repeated applications, which is the
operational motivation for noiseless codes (Definition \ref{def:noiseless}).

Noiseless and unitarily noiseless information are preserved
indefinitely.  No matter how many times $\cE$ is applied, we can still
distinguish code states.  In classical information theory, the number
of errors ({\em i.e.}, bit flips) required to transform one code word
into another is called the \emph{distance} of the code.  Under the
more general definition of distance introduced by Knill {\em et al.}
\cite{KnillPRL00} (based on defining a single application of $\cE$ as
an ``error''), noiseless and unitarily noiseless codes are
\emph{infinite-distance codes}, with respect to the noise model
defined by $\cE$.  Each infinite-distance code is a manifestation of
an underlying noiseless or unitarily noiseless IPS.  Infinite-distance
IPSs may be viewed as degrees of freedom into which $\cE$ introduces
no entropy at all, transforming them reversibly (if at all).  We do
not have to pump entropy out of infinite-distance IPS, and so no
active error correction is required.  For this reason, these have also
been called \emph{passive} error-correcting codes.

%*********************************************************************
\subsection{Constraints on the recovery operation}

Suppose that we \emph{can} do something to the system \emph{in
between} applications of $\cE$. This is crucial whenever the channel
preserves information, but maps it to a part of the Hilbert space that
is unprotected against further applications of $\cE$.  Now we must
intervene, applying active correction to move our precious information
back into protected sectors, and ensure its continued survival.  If we
can do absolutely anything, then we can correct any preserved code
(thanks to Theorem \ref{thm:PresCorr}).  In practice, however, we may
only be able to do certain operations.  Any CPTP map can be decomposed
into (i) a POVM measurement, followed by (ii) a conditional unitary
that depends on the outcome of the POVM.  This decomposition suggests
two natural restrictions on $\cR$: It can consist only of a
measurement, or it can be completely unitary.

\subsubsection{Measurement-stabilized codes}

If unitary operations are costly or noisy, but measurements can be
performed relatively quickly, the only ``corrections'' that we can
perform effectively are pure measurements.  For our purposes
\footnote{This careful definition may seem pedantic.  However,
``measurements'' are sometimes defined very generally, with an update
rule involving \emph{any} square root of the effect $E_m$.  This
trivializes our distinction between measurements and arbitrary
CP-maps.  The convention we adopt here is known as \emph{L\"uder's
Rule}, and defines the unique minimally disturbing (and maximally
repeatable) implementation of a given measurement.}, a
measurement is a POVM defined by a set of effects,
\begin{equation*}
\cM = \{E_m\},\;\;\mathrm{\ where\ }\sum_m{E_m = {\Id}}.
\end{equation*}
The outcome of such measurement is a particular value of $m$, with
probability $Pr(m) = \Tr(E_m\rho)$, and a post-measurement state 
$$\rho \to E_m^{\frac12}\rho E_m^{\frac12},$$, 
\noindent 
where $E_m^{\frac12}$ is the unique positive semidefinite square root
of $E_m$.

Can we use measurements to correct noise?  At first, it seems
implausible -- after all, while a measurement provides information, it
actually does not \emph{do} anything.  However, the existence of
unitarily noiseless codes shows that passive information gain, such as
knowing how many times $\cE$ has been applied, can be useful.  This
motivates a definition of \emph{measurement-stabilized codes}, whose
information is preserved indefinitely provided that a measurement is
performed after every application of the channel:

\begin{definition}
\label{def:MS}
A code $\cC$ is \textbf{measurement-stabilized} for a CPTP map $\cE$
if there exists a measurement $\cM=\{E_m\}$ such that, conditional on any
outcome $m$, $\cC$ is unitarily noiseless for $\cM\circ\cE$.
\end{definition}

Stabilizer codes for Pauli channels \cite{NielsenBook00} are an
example of measurement-stabilized codes.  Stabilizer codes divide the
system into two degrees of freedom, the code and the syndrome.
Measuring the syndrome ``collapses'' the error, revealing which Pauli
unitary transformed the information-carrying subsystem.  In the usual
paradigm, we would undo this unitary -- but this is not actually
necessary, as long as we keep track of the current ``Pauli frame''
\cite{Knill05} by recording the results of each syndrome measurement
as the system evolves.

The key to reconciling the behavior of stabilizer codes with
Definition \ref{def:MS} is conditioning on the syndrome measurements.
Since each syndrome measurement collapses the syndrome subsystem into
a particular basis state, we can see the overall system's dynamics,
conditional on the measurement record, as a rather strange
time-dependent unitary evolution: At each time step, the code subspace
gets transformed by some Pauli operator $P_l$, and the syndrome state
jumps from $\ket{k}\to\ket{k+l}$.  Since the code evolves unitarily at
every step, it is unitarily noiseless, and the information in it can
be recovered at any time.

At first glance, this may seem trivial, for as we observed above,
\emph{any} correction operation $\cR$ can be written as a measurement
followed by a conditional unitary.  So, given a generic correctable
code, couldn't we just do the measurement, skip the conditional
unitary, and keep track of which unitary we did not do?  This does not
work in general, because $\cE$ may have moved the code to a different
subspace which is not, itself, a code.  Stabilizer codes can be
measurement-stabilized because they actually comprise a large set of
preserved codes, and (conditional on the syndrome measurement) the
channel merely permutes the codes while transforming them unitarily.
It is an open question whether all measurement-stabilized codes are of
this form (that is, a large set of isomorphic codes, indexed by a
syndrome), or if the above definition permits other structures.

\subsubsection{Unitarily correctable codes}

In some systems, we have the opposite situation: Measurements are slow
and/or hard, while unitary evolution is fast and relatively easy
(liquid-state NMR quantum computation is an extreme example; most
solid-state architectures also fall into this category).  Now we can
only apply unitary gates after each application of $\cE$.  The authors
of Ref. \onlinecite{OQECC} considered this situation, and demanded that
there exist a unitary matrix $U$ on $\cH = (\cH_A\otimes \cH_B) \oplus
\cH_C$ such that $\tr_B \{U \cE(\rho_{AB})U^\dagger \} = \tr_B
\rho_{AB}$ for all $\rho_{AB} \in \cB(\cH_{A}\otimes\cH_B)$.  The $A$
subsystem is a \emph{unitarily correctable}\footnote{The authors of
\cite{OQECC} called this ``unitarily noiseless'', but we believe the
term ``unitarily correctable'' is more appropriate.} subsystem (see
also \cite{KS06a}).

\begin{definition} 
\label{def:UC}
A code $\cC$ is \textbf{unitarily correctable} for a channel $\cE$ if
there exists a unitary correction map $\cU(\cdot)= U\cdot U^\dagger$,
for some unitary operator $U$, so that $\cC$ is noiseless for
$\cU\circ\cE$.
\end{definition}

Unitarily correctable codes are interesting in part because $\cE$ does
not inject entropy into the code states\footnote{Actually, it is
slightly more technical than this: Given any unitary correctable code,
there is another code associated with the same unitarily correctable
IPS, into which $\cE$ does not inject any entropy.  This is directly
related to the fact that a code can be noiseless without being fixed
-- in both cases, repeated application of $\cE$, or $\cU\circ\cE$,
causes the code to converge toward a fixed code, whose entropy does
not increase thereafter.}.  If it did, the error could not be
corrected by a unitary operation.  Kribs and Spekkens considered
unitarily correctable codes in some detail in Ref. \onlinecite{KS06a}, and
noted that while any preserved code is ``unitarily recoverable'' --
{\em i.e.}, there is a unitary that puts the information back into the
subsystem where it originated -- this need not suffice to correct the
errors, and cooling may be required to protect the information against
subsequent iterations of the noise.  Sufficient conditions for unitary
correctability have likewise been directly derived from the structure
of $1$-isometric encodings \cite{Ticozzi2010} (see, in particular,
Prop. 1 therein).

\begin{example}
Consider two qubits labeled $A$ and $B$, and let $\cE$ act as follows:
%\begin{enumerate}
$B$ is measured in the $\{\ket{0},\ket{1}\}$ basis; If the result was
``1'', then $A$ is depolarized.  Finally, $B$ is depolarized.
%\end{enumerate}
The code $\cC = \{\rho_A\otimes\proj{0}_B\ \forall\ \rho_A\}$ is
preserved.  It is unitarily recoverable -- in fact, no recovery is
necessary because the information remains in the $A$ subsystem.  It is
{\em not} unitarily correctable, however, because unless $B$ is cooled
to the $\ket{0}$ state, $\cE$'s next iteration may damage the
information.
\end{example}

Kribs and Spekkens also pointed out that, under certain circumstances,
unitarily correctable codes can be found efficiently.  This
observation is closely related to our next topic.
%{unconditionally preserved} information.

%*********************************************************************
\subsection{Unconditionally preserved information}

If a code $\cC$ is preserved, then Bob can distinguish between states
in $\cC$ (and their convex combinations) just as well as Alice.  So if
we want to know ``Was the system prepared in $\ket{\psi}\in\cC$?'',
Bob can answer
%\footnote{Throughout this section, statements like "Bob can determine X'' should be interpreted as ``Bob can determine X just as well as Alice could have.''  For instance, neither Bob nor Alice can determine with certainty whether a system was prepared in $\ket{\psi}$ -- but Alice can measure $\cM = \{\proj{\psi},\Id-\proj{\psi}\}$, and the pertinent question is whether Bob can simulate that measurement.}
just as well as Alice could have, by discriminating $\proj{\psi}$ from
a convex combination of all states orthogonal to $\ket{\psi}$.  What
he cannot do is determine whether the initial state was in $\cC$.
Information is preserved \emph{conditional} on the system being
prepared in $\cC$, as illustrated by the following example.

\begin{example} 
\label{ex:UCP1}
Let $\cE$ be the following [effectively classical] channel from a
$d$-dimensional system to itself.  On the subspace $\cH_{d-1}$ spanned
by $\{\ket{0}\ldots\ket{d-2}\}$, $\cE$ acts as the identity channel.
However, $\ket{d-1}$ is decohered and mapped to the maximally mixed
state $\frac{1}{d}{\Id}$.

The code comprising all states on $\cH_{d-1}$ is preserved, so Bob can
distinguish between $\ket{0}$ and any convex combination of
$\ket{1}\ldots\ket{d-2}$.  If the input state was supported on
$\cH_{d-1}$, Bob can determine whether $\ket{0}$ was prepared.
Without this promise, however, \emph{any} measurement result on the
output is consistent with the input state $\ket{d-1}$.
\end{example}

%Preserved codes are subsets of the state space, and their information is preserved \emph{conditional} upon one overarching promise:  that the input state really was one of the code states.  This promise represents some prior certainty about the input, which Bob can leverage into more certainty.  Absent such a promise, it may be impossible to recover any property of the initial state -- i.e., for Bob to determine (just as well as Alice could have) the truth or falsehood of any logical proposition about the input.  A proposition is a subspace $\cP\subseteq\cH$, and it is true if and only if $\rho$ is supported entirely on $\cP$.

%\begin{example}
%Consider the channel of Example \ref{ex:ClassicalFourState}, represented by the stochastic map
%$$\cE = \begin{pmatrix} \hf&0&0&\hf \\ \hf&\hf&0&0 \\ 0&\hf&\hf&0 \\ 0&0&\hf&\hf \end{pmatrix}.$$
%This channel preserves information, because the codes $\cC = \{0,2\}$ and $\cC'=\{1,3\}$ are each preserved.  However, unless we're promised that the initial state was selected from $\cC$ (respectively, $\cC'$), there is no property of the initial state that can be inferred with certainty.  For example, if the output is $1$, then we can infer that the input state was in the subset $\{0,1\}$, and \emph{not} in $\{2,3\}$.  If, on the other hand, the output is $2$, then the input might have been $1$ or $2$, so either proposition might be true or false.  For each nontrivial proposition $\cP$, there is a possible output state for which we cannot determine the truth of $\cP$.
%\end{example}

Sometimes, a channel preserves some properties of the input state
\emph{irrespective} of what it is.  For instance, if $\cE$ is the
identity channel, then Bob can make any measurement that Alice can.
His conclusions from those measurements do not depend on any prior
information about the input.  The following example is less trivial.

\begin{example}
\label{ex:UnconditionalPreservation}
Consider the classical channel whose action is pictorially shown below:
\begin{center}
\includegraphics[width=35mm]{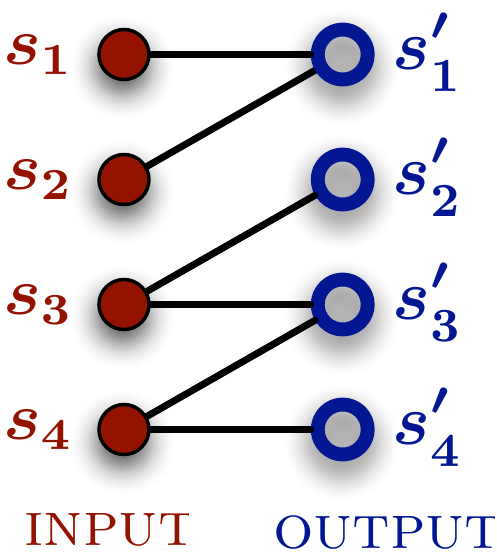}
\end{center}
which corresponds to a stochastic map of the form
$$\cE = \begin{pmatrix} 1&1&0&0 \\ 0&0&\frac12&0 \\
0&0&\frac12&\frac12 \\ 0&0&0&\frac12 \end{pmatrix}.$$
\noindent 
Bob can measure $\{1'\}$ vs. $\{2',3',4'\}$, and from the result infer
exactly what Alice would have gotten had she measured $\{1,2\}$
vs. $\{3,4\}$.  So this property of the input state is unconditionally
preserved: No matter what the input state was, Bob can determine
whether it was in $\{1,2\}$ or not.  Note that unconditional
preservation need {\em not} be related to noiselessness -- applying
this channel twice ruins the information.
\end{example}

This illustrates \emph{unconditionally preserved information}.  The
most natural way to define unconditional preservation is not in terms
of states or codes, however, but rather in terms of measurements.

\begin{definition}
Let $\cE:\cB(\cH)\to\cB(\cH')$ be a channel, and $\cM = \{P_1\ldots
P_n\}$ a projective measurement on Hilbert space $\cH$ (so
$\sum_k{P_k}={\Id}$).  Then $\cM$ is \textbf{unconditionally
preserved} by $\cE$ if and only if there exists another measurement
$\cM' = \{Q_1\ldots Q_n\}$ on $\cH'$ such that $\cM'$ simulates $\cM$:
that is, $\Tr[P_k\rho] = \Tr[Q_k\cE(\rho)]$ for all density matrices
$\rho$ on $\cH$.
\end{definition}

This condition on measurements is based in the Heisenberg picture of
quantum mechanics, in which states stay fixed, but measurements evolve
according to $\cE^\dagger$.  In order for $\cM$ to be unconditionally
preserved, there must be some measurement $\cM'$ that evolves into
$\cM$.  We can also if desired define an equivalent condition on states:
%(though the Heisenberg picture is probably more fundamental in this case).

\begin{definition}
A code $\cC$ is \textbf{unconditionally preserved} by a channel $\cE$
if and only if the Helstrom measurement for every weighted pair of
states $p\rho,q\sigma$ in the convex closure of $\cC$ is
unconditionally preserved.
\end{definition}

The second definition is strictly more general: Every unconditionally
preserved measurement $\cM = \{P_1\ldots P_2\}$ can be identified
uniquely with a code
$$\cC =
\left\{\frac{P_1}{\Tr(P_1)}\ldots\frac{P_n}{\Tr(P_n)}\right\},$$ which
is unconditionally preserved if and only if $\cM$ is.  Every classical
code whose support is all of $\cH$ defines a single unconditionally
preserved measurement.  Quantum (or hybrid) codes whose support is all
of $\cH$ define entire algebras of unconditionally preserved
measurements.  Codes restricted to a subspace do not generally
correspond to unconditionally preserved measurements.

The code associated with a given unconditionally preserved measurement
spans the entire Hilbert space.  Therefore, following the proof of
Theorem \ref{thm:PresCorr}, it can be corrected using a transpose map
$\cEP$ -- where $\cP$ is the entire Hilbert space!  Since this
statement holds for every unconditionally preserved measurement, we
can correct \emph{every} unconditionally preserved code using a single
unique recovery, which we denote $\hat\cE$:
\begin{equation}
\hat\cE(\cdot) = \cE^\dagger\left( \cE({\Id})^{-\frac12} \cdot
\cE({\Id})^{-\frac12} \right).
\end{equation}
It follows that every unconditionally preserved measurement consists
of projectors $P_k$ that are fixed points of $\hat\cE\circ\cE$.  There
exists a \emph{unique} unconditionally preserved IPS, which contains
all the unconditionally preserved codes.  Moreover, we can find its
structure quite easily by constructing and diagonalizing
$\hat\cE\circ\cE$.  Other codes are hard to find, precisely because we
need to know their support $\cP$.

Kribs and Spekkens observed that if $\cE$ is unital (that is,
$\cE({\Id}) = {\Id}$), then its unitarily correctable
codes are fixed points of $\cE^\dagger\cE$.  This is an interesting
special case of unconditional preservation.  If $\cC$ is unitarily
correctable, then the channel does not add any entropy to it -- thus,
every pure state in the code remains pure.  But if $\cE$ is unital, it
cannot map two orthogonal subspaces to overlapping subspaces of the
same size, because this would cause a pile-up of probability on the
overlapping portion.  So every unitarily correctable code must be
unconditionally preserved, because no other subspace can be piled on
top of it in the output space.  Finally, for a unital channel,
$\hat\cE = \cE^\dagger$, so $\cE^\dagger$ corrects every
unconditionally preserved code.

%******************************************************************************
%******************************************************************************
\section{Applications}
\label{sec:App}

In this section, we present three applications of the IPS framework
that we have derived.  First, we state a very simple algorithm that
efficiently finds all noiseless and unitarily noiseless codes for a
given map $\cE$.  We then present a similar algorithm to find all the
unconditionally preserved codes.  Finally, we show how to address
so-called ``initialization-free'' DFSs and NSs within our framework.

%****************************************************************************
\subsection{Finding infinite-distance codes}

Our discussion suggests a natural strategy for finding all the
preserved codes of a channel $\cE$: First, find all its preserved IPS;
then build codes from the IPS.  Unfortunately, there is a potential
IPS for each and every subspace $\cP\subseteq\cH$.  So, searching for
IPS seems to require an exhaustive search over all subspaces of $\cH$
(see \cite{ChoiPRL06}).  We can find \emph{some} preserved codes by
picking particular subspaces, but we may not find the largest IPS (or
any of them).  Since the problem is NP-hard, an efficient algorithm
seems unlikely (though it should be noted that we have only proven
that finding the best \emph{classical} code is NP-hard -- other
special cases, for instance the largest quantum code, might
conceivably be easier).

Let us focus instead on \emph{noiseless} codes.  The noiseless IPS of
$\cE$ is unique, because all the maximum noiseless codes are isometric
to $\cE$'s fixed points.  So, to find the unique noiseless IPS, we
need only determine the structure of $\cE$'s fixed points.  Theorem
\ref{thm:FixedPtThm} defines this structure, and suggests an efficient
algorithm to find it:\\

\noindent{\bf Algorithm for finding noiseless IPS:}
\begin{enumerate}
\item Write $\cE$ as a $d^2\times d^2$ matrix, where $d$ is the
dimension of the Hilbert space.

\item Diagonalize the matrix, and extract its eigenvalue-1 right and
left eigenspaces (corresponding to $\FixE$ and $\FixEdag$,
respectively).

\item Compute $\cP_0$, the support of $\FixE$, and project $\FixEdag$
onto $\cP_0$ to obtain a basis for $\cA$.

\item Find the shape of $\cA$.
\end{enumerate}

In the last step, we need to find the canonical decomposition, Eq.
\eqref{decomp} of a finite-dimensional matrix algebra specified as a
linear span. This can be done efficiently using, for example, the
algorithm presented in Ref. \onlinecite{HolbrookQIP04}. This canonical
decomposition step is also present in existing algorithms for finding
NSs \cite{ChoiPRL06,Knill06a}.  Our algorithm improves on previous
algorithms by providing a straightforward method of finding $\cA$ as a
linear span.  Its hardest step is diagonalizing a $d^2\times d^2$
matrix, which runs in time $O(d^6)$. As such, it is more efficient
than algorithms (such as \cite{ZurekPTP93,ChoiPRL06}) that require
exhaustive search over states or subspaces in $\cH$, for these sets
grow exponentially in volume with $d$.

We can generalize this algorithm to find an arbitrary channel's
unitarily noiseless IPS.  Whereas the noiseless IPS consists of
$\cE$'s fixed points -- operators $X$ such that $\cE(X) = X$ -- the
unitarily noiseless IPS consists of \emph{rotating points} --
operators $X$ such that $\cE(X) = e^{i\phi_X}X$.

\begin{definition}
\label{def:RotatingPoint}
Let $\cE:\cB(\cH)\to\cB(\cH)$ be a CPTP map.  An operator
$X\in\cB(\cH)$ is a \textbf{unitary eigenoperator} of $\cE$ if and
only if $\cE(X) = e^{i\phi}X$ for some $\phi\in\bR$.  The
\textbf{rotating points} of $\cE$ comprise all operators in the span
of its unitary eigenoperators.
\end{definition}

Note that a rotating point need not be an eigenoperator -- for
instance, a linear combination of two unitary eigenoperators with
different phases is a rotating point, but not itself an eigenoperator.
As an example, consider the unitary qubit channel $\cE(\rho) =
e^{-i\phi\sigma_z} \rho e^{i\phi\sigma_z}$.  The Pauli operators
$\sigma_x$ and $\sigma_y$ are not eigenoperators, but they are
rotating points.

\begin{numlemma}{\ref{lem:UNRot}}
\LEMUNROT
%\todo{Compare to older version:}
%Every unitarily noiseless code for $\cE$ is isometric to the set of all 
%(positive trace-1) states in the span of the rotating points of $\cE$.
\end{numlemma}
%\begin{proof} See Appendix \ref{sec:Proofs}. \end{proof}

We adapt the above algorithm by shifting its focus from fixed points
to rotating points.  It is useful to note that the support of the
rotating points is the same as the support $\cP_0$ of the fixed
points.  Therefore, we just need to replace step 2 above by the
following:

\begin{itemize}
\item[$2'$.] Diagonalize the matrix, and extract the right and left
eigenoperators with unit modulus eigenvalues. Let $\FixE$ ($\FixEdag$)
be the linear span of the unit-modulus right (left) eigenoperators.
\end{itemize}

\noindent 
This again runs in time $O(d^6)$ as before.  It is (to our knowledge)
the first efficient algorithm to find unitarily noiseless codes for
arbitrary channels.  We note in passing that both algorithms -- for
finding noiseless and unitarily noiseless IPS -- rely on the codes
having infinite distance, so they are unlikely to be adaptable to
finding other kinds of IPS.

%****************************************************************************
\subsection{Finding the unconditionally preserved IPS}

We know that preserved codes are in general hard to find, but in the
previous section we saw how to take advantage of infinite-distance
codes' structure to find the unique noiseless and unitarily noiseless
IPSs.  Unconditionally preserved IPSs are another special case.  A
channel has a unique unconditionally preserved IPS, and we can find it
efficiently.  The algorithm is extremely simple: Construct
\begin{equation}
\hat\cE(\cdot) = \cE^\dagger\left( \cE({\Id})^{-\frac12} \cdot
\cE({\Id})^{-\frac12} \right),
\label{uncondrec}
\end{equation}
diagonalize $\hat\cE\cE$, and extract its fixed points (the eigenspace
with eigenvalue $+1$).  These will form an algebra, which defines the
IPS we are looking for.  However, one might reasonably inquire why the
unconditionally preserved IPS is interesting and useful.

If we ask ``What information is preserved by a given channel $\cE$?'',
then one possible answer consists of an exhaustive list of all the
channel's preserved IPSs.  This is somewhat unsatisfactory for three
reasons.  First, we do not know how to find such a list (though we
know that it is generally hard).  Second, it might be very very long,
even for channels on small systems.  Third, the preserved codes
corresponding to these IPSs represent information that \emph{could} be
preserved by the channel, depending on what the sender chooses to do,
and conditional upon prior agreement between sender and receiver.

Unconditional preservation provides an alternative answer.  Every
channel has a unique unconditionally preserved IPS, comprising all the
information that is \emph{definitely} preserved by $\cE$.  In the
important case where the ``sender'' is a natural process, this IPS
represents everything that the observer can determine with certainty.
Any further conclusions are valid only conditional upon certain prior
assertions about the ``distant'' system ({\em e.g.}, that its state
lay in some subspace $\cP$).  This interpretation alone is sufficient
reason to consider the unconditionally preserved IPS -- independent of
the happy accident that it is unique and easily calculable.

%************************************************************************
\subsection{Initialization-free DFS and NS}

As discussed in Section \ref{sub:kinds}, DFSs and NSs are
manifestations of $\cE$'s noiseless IPS.  We can demand further
operational requirements on a DFS or NS.  One particular criterion is
{\em robustness} against initialization errors -- that is, we demand
not only that information encoded \emph{in} the DFS/NS be preserved
indefinitely, but also that if Alice failed to prepare a state within
the DFS/NS, that this can be detected by Bob.  Such
``initialization-free'' (IF) DFS and NS were first studied in Ref.
\onlinecite{ShabaniPRA05}, and have been further characterized in
Ref. \onlinecite{TicozziTAC} in the context of Markovian dynamics.
Since a DFS is just a NS with a trivial noise-full subsystem, we shall
focus on IF-NS\footnote{Ref. \onlinecite{ShabaniPRA05} actually discusses
a more general case, allowing unitary evolution of the NS.  This
is what we call a unitarily noiseless code. To be consistent with the usual
definition of NS, we use the ``strict'' NS condition given in Eq.
\eqref{defNS2}, but everything in this section can easily be generalized
by using a channel's unitarily noiseless IPS instead of its noiseless IPS.}.

If we decompose the system's Hilbert space as
$$\cH=(A\otimes B)\oplus C,$$
\noindent 
and $A$ supports a NS, then we can write an arbitrary density
operator in the following block form:
\begin{equation}
\label{ImpInit}
\rho=\left(\begin{array}{cc}\rho_{AB}&\bar\rho\\
\bar\rho^\dagger&\rho_C\end{array}\right).
\end{equation} 
The NS is said to be perfectly initialized whenever $\bar\rho$ and
$\rho_C$ are zero.  If, in practice, it is not possible to guarantee
preparation within $A\otimes B$, then we need a special kind of NS that
is insensitive to such initialization errors.  The NS is
\emph{initialization-free} if the (possibly subnormalized) state
$\rho_{AB}$ on $A\otimes B$ satisfies the NS condition of
Eq. \eqref{defNS2}, even when $\rho_C$ is not zero.  In other words,
an IF-NS is one that is \emph{immune to interference} coming in from
orthogonal subspaces of $\cH$ (i.e., states that would not have been
prepared if the system had been perfectly initialized).

Our framework, as it turns out, provides a simple and elegant condition
for initialization-free NSs:  \emph{An NS is IF if and only if it is
noiseless \emph{and} unconditionally preserved}.  So, we can find a
channel's IF noiseless structures by intersecting its noiseless IPS
and its unconditionally preserved IPS.  In the remainder of this section,
we will demonstrate this equivalence.

Given a Hilbert space $\cH$ and a channel $\cE$, the channel's noiseless
IPS defines a subspace decomposition $\cH=\cP_0\bigoplus\overline\cP_0$.
Subspace $\cP_0$ is the support of the noiseless IPS.  The noiseless IPS
also defines a canonical decomposition of $\cP_0$ into $k$-sectors 
($A_k\otimes B_k$), so we write the Kraus operators of $\cE$ accordingly, as:
\begin{equation}\label{eq:IFNSKraus}
K_i=\left(\begin{array}{cc}\sum_k{\Id}_{A_k}\otimes\kappa_{i,B_k} &D'_i
\\ 0&C'_i\end{array}\right).
\end{equation}
Each $k$-sector is an invariant subspace.  So each NS ($A_k$) is automatically
resilient to initialization errors that prepare states in the wrong $k$-sector
(but still within $\cP_0$).

However, if faulty initialization puts support on $\overline\cP_0$, then this
error may spill into the noiseless sector.  Specifically, the $D'_i$ blocks in
Eq. \ref{eq:IFNSKraus} map $\overline{\cP_0}$ into $\cP_0$, which can interfere
with information stored in noiseless codes.  Since every NS is immune to
interference from other $k$-sectors within $\cP_0$, let us consider interference
from $\overline{cP_0}$.  

Consider, for the sake of simplicity, a noiseless IPS
containing a single $k-sector$, so $\cP_0 = A\otimes B$
(as in Eq. \eqref{ImpInit}).  Let $P$ be the projector onto $\cP_0$.
The Kraus operators are 
\begin{equation}
K_i=\left(\begin{array}{cc}A_i&D_i\\0&C_i\end{array}\right),
\end{equation}
and if the initial state is $\rho$ as given in Eq. \eqref{ImpInit},
then the final state on $\cP_0 = A\otimes B$ is
\begin{eqnarray}
P\cE(\rho)P&=&\sum_i A_i\rho_{AB}A_i^\dagger+\sum_iD_i\rho_C
D_i^\dagger\nonumber\\ &&+\sum_i(A_i\bar\rho
D_i^\dagger+D_i\bar\rho^\dagger A_i^\dagger).
\end{eqnarray}
For perfect initialization, only the first term is present.  The
remaining terms represent interference from faulty initialization
on $\overline{\cP_0}$.  The NS is IF if and only if they vanish,
which requires
\begin{equation}
\label{IFcond}
\sum_iD_i\rho_C D_i^\dagger=-
\sum_i(A_i\bar\rho D_i^\dagger+D_i\bar\rho^\dagger A_i^\dagger).
\end{equation}
Since $\rho_C$ is positive semi-definite, the left-hand side of Eq. \eqref{IFcond}
is also positive semidefinite.  But the right-hand side of Eq. \eqref{IFcond}
must be traceless, because in order for $\cE$ to be trace-preserving, 
$\sum_iA_i^\dagger D_i=0$, and so
\begin{equation}
\Tr\left[
\sum_i(A_i\bar\rho D_i^\dagger+D_i\bar\rho^\dagger A_i^\dagger)\right] =
2\Re\Tr\left(\sum_i A_i^\dagger D_i\bar\rho^\dagger\right) = 0.
\end{equation}
So the left-hand side is positive semidefinite \emph{and} traceless, which means it
vanishes -- and so Eq. \eqref{IFcond} holds if and only if
$\sum_iD_i\rho_C D_i^\dagger=0$ for all $\rho_C$ -- which implies $D_i=0$ for all $i$.

This means that in order for an NS whose support is $\cP_k = A_k \otimes B_k$ to be 
IF, the channel must not map anything from $\overline \cP_0$ into $\cP_k$.
That is, $\cP_k$ is orthogonal to $\cE(\rho_C)$ for every
$\rho_C\geq0$ on $\overline\cP_0$ (and, by Lemma \ref{lem1:Subsp} in Appendix \ref{sub:preserved},
it is sufficient to consider just one full-rank $\rho_C$ on $\overline{\cP_0}$). 
But this is precisely the condition for the corresponding code to be unconditionally
preserved: Bob must be able to determine whether the system was
correctly initialized, which means that the channel must not map any
part of $\overline{\cP_0}$ back into $\cP_k$.

%******************************************************************
%******************************************************************
\section{Conclusions and Outlook}
\label{sec:Conc}

We have presented a framework characterizing the information preserved
by a quantum process, described by an arbitrary CPTP $\cE$ map acting
on a finite-dimensional quantum system.  Information is carried by
codes; codes are preserved if their associated information can be
extracted after passing through the channel; preservation implies
correctability.  Preserved codes are built upon the channel's
information preserving structures (IPSs), which in turn inherit matrix
algebra structure from fixed point sets of CPTP maps.  This allows for
a very elegant and concise description of the full
information-carrying capability of any code.  We also discussed
several operational variations on preservation, with particular
attention to infinite-distance codes, and applied the theory to find
all of a channel's noiseless, unitarily noiseless, and unconditionally
preserved codes.

A number of important open problems and directions for further
investigation remain.  We have not explicitly addressed continuous-time
quantum processes.  Such a process is described by a 1-parameter family
$\{\cE_t: t\geq 0\}$ of CPTP maps.  A special subclass with particular physical
significance is Markovian noise, where $\cE(t) = e^{t\cL}$ for some
Liouville semigroup generator $\cL$ \cite{Alicki}.  In principle, our definitions of
noiseless and unitarily noiseless codes extend to the Markovian
setting, suggesting connections to recent studies of DFSs/NSs
under Markovian noise (see in particular Refs.
\onlinecite{ShabaniPRA05,TicozziTAC,Automatica,Oreshkov2010}), and to
earlier approaches such as ``damping bases'' developed in the
context of quantum optics \cite{DampBases}.  However, we believe it will
be necessary to extend our notion of correctability to address continuous-time
QEC, as developed for instance in \cite{Ahn}.

Our analysis has focused on information preservation under the uncontrolled
(``free'') evolution of an open system.  The ability to control that system's
dynamics \emph{while} it is experiencing noise (rather than correcting the errors
after they occur) raises questions that are interesting for practical quantum
information processing and from a control-theoretic perspective. It would be
valuable to know how to synthesize dynamics that support a given (desired) IPS,
using externally applied control, much as DFSs/NSs can be engineered using 
open-loop unitary manipulations \cite{dygen} or closed-loop feedback protocols
\cite{TicozziTAC,Automatica}.

Our current framework does not address ``post-selective'' preservation of
information, where the information is preserved conditional on a particular
measurement outcome.  Another natural direction for
generalization is to relax the ``zero-error'' requirement, looking at
imperfectly preserved information under CPTP channels or more general
noisy dynamics.  Preliminary investigations \cite{NgBK} indicate that
partial extensions of some of the structures present in the perfect
case carry over to the approximate case, but a variety of interesting
complications arise.  A final question that deserves further investigation
arises when the information-carrying system is not \emph{initially}
fully decoupled from its environment.  This particular kind of initialization
error can produce noise which cannot be described by CP maps, and
its analysis must address the influence of (weak) initial correlation with
the environment on the information [supposedly] stored within the system.

%*******************************************************************
%*******************************************************************
\acknowledgments
%% LV: Pl check for accuracy!

The authors acknowledge support in part by the Gordon and Betty Moore
Foundation (DP and HKN); by NSERC and FQRNT (DP); by the NSF under
Grants No. PHY-0803371, No. PHY-0456720; No. PHY-0555417 and
No. PHY-0903727 (LV); and by the Government of
Canada through Industry Canada and the Province of Ontario through
the Ministry of Research \& Innovation (RBK).  We also gratefully
acknowledge extensive conversations with Robert Spekkens, Daniel
Gottesman, Cedric Beny, and Wojciech Zurek.

%******************************************************************
%******************************************************************
\begin{appendix}

\section{Our framework for analyzing information}

\subsection{Our notion of information: Relation to Shannon theory}
\label{sec:ShannonTheory}

The most common technical meaning of ``information'' comes from
Shannon's theory of communication
\cite{ShannonBSTJ48,ShannonBook49,CoverBook91}.  Here, Alice and Bob
are connected by a communication channel $\mathcal{E}$ (a dynamical
map between input states and output states), and also have:

\begin{enumerate}
\item A codebook that tells Bob which signals Alice might send;
\item The patience and ability to send signals requiring arbitrarily
many uses of the channel;
\item A willingness to tolerate a very small probability of failure;
\item A guarantee that $\mathcal{E}$ will be applied exactly once.
\end{enumerate}
Although this paradigm is the backbone of both classical and quantum
information theory, it is not unique.  Any or all of the above
resources may be unavailable:

\begin{itemize}
\item Sometimes there is no codebook restricting the possible signals.
In scientific applications, the source of information is generally a
natural phenomenon rather than a canny and cooperative sender.  This
\emph{observational} paradigm restricts the questions whose answers
the receiver can learn.
\item In real-time applications, a signal has to be transmitted within
a strictly limited number ($N$) of channel uses.  This eliminates the
second resource (encoding over arbitrarily many uses), and motivates
\emph{single-shot} capacity: What can we accomplish with a single use
of the channel $\mathcal{E}^{\otimes N}$?
\item Some applications demand perfect reliability.  This eliminates
the third resource (tolerance of arbitrarily small failure
probability), and yields \emph{zero-error information theory}
\cite{ShannonITIT56,KornerIEEE98}.
\item Memory devices, which store information rather than transmitting
it, may violate the guarantee that $\cE$ is applied exactly once.  We
may wish our information to be preserved for an arbitrary number of
clock cycles, or $\mathcal{E}$ may be a snapshot of a continuous
process.  When $\cE$ may be applied many times, we turn to \emph{error
correction}.  \emph{Correctible} information requires active
correction after each iteration of $\mathcal{E}$; \emph{noiseless}
information persists through repeated iterations of $\mathcal{E}$ with
no intervention.
\end{itemize}

In this paper, we are concerned primarily with identifying the kinds
of information that can be preserved, rather than the rate at which
information can be sent or stored.  So, we focus on zero-error
information and the single-shot paradigm.  This does not really affect
the generality of our results: Since they apply to arbitrary channels,
we can discuss $\mathcal{E}^{\otimes N}$ for any $N$.  We do not know
for certain, however, whether tolerating an asymptotically small
amount of error changes the \emph{kinds} of information that can be
preserved by $\mathcal{E}^{\otimes N}$.

The other two resources (a pre-existing codebook, and exact knowledge
of $\mathcal{E}$) are quite important.  They yield different
preservation criteria, with substantially different consequences, and
we consider them separately.

\subsection{On the usefulness and generality of codes}
\label{sec:CodeApology}

Our framework for analyzing preserved information relies on
\emph{codes} to describe different kinds of information.  A code is an
arbitrary set of preparations (states) for a physical system $\cS$,
representing the alternatives available to the sender.  Essentially, a
code describes a very generalized ``subsystem'', in which information
can be encoded.  We settled on this formalism after quite a bit of
thought and exploration, and expect that some readers may seek a more
extensive explanation of why we believe it is useful, general, and
powerful.  The most efficient way to do so might be to anticipate some
potential objections.

\begin{itemize}
\item \emph{Using ``questions'' to define information seems inherently
classical, and inadequate to describe quantum information.}  The idea
of a question, with a definite answer, is indeed inherently classical.
Human beings are unavoidably classical, and as Bohr famously insisted
\cite{WheelerBook83}, our descriptions and perceptions of Nature are
always classical.  As such, we believe that a precise and general
definition of ``information'' must rely on classical concepts.  We can
nonetheless describe quantum information in this framework.  The
difference between a classical bit and a quantum bit is that the bit
admits just one sharp question, ``Is the bit 0 or 1?,'' whereas the
qubit supports an infinite continuum of inequivalent sharp questions,
``Is the qubit in state $\ket\psi$ or state $\ket{\psi_\perp}$,'' for
every orthogonal basis $\{\ket\psi,\ket{\psi_\perp}\}$.  By using
classical questions as a common denominator to define both classical
and quantum information in the same lingua franca, we have a framework
that is open to novel forms of information -- rather than begging the
question of whether they exist.

\item \emph{This definition does not seem to capture entanglement as a
form of information -- i.e., that $\cE$ might preserve entanglement
between $\cS$ and a reference system $\cR$.}  Entanglement is a
peculiarly quantum form of correlation, wherein the state of $\cS$ is
conditional upon observations on the reference system.  Projecting
$\cR$ into a state $\ket\psi$ \emph{steers} \cite{Steering} $\cS$ into
a corresponding $\rho_\psi$.  It is not difficult to show that $\cE$
preserves this entanglement if and only if it also preserves the code
comprising \emph{all} $\rho_\psi$ into which $\cS$ can be steered.
Thus, the code paradigm does address entanglement as a form of
information.

\item \emph{Preserved information should be addressed in the
Heisenberg picture, by considering preserved observables rather than
states.}  In fact, our analysis proceeds along these lines; we demand
that every measurement for distinguishing between code states be
reproducible on Bob's end.  However, the code $\cC$ is a crucial
ingredient in defining a kind of information, because it determines
which measurements need to be reproducible!  Otherwise, it is easy to
identify \emph{all} POVMs that can be reproduced on Bob's end with
``preserved information'' \cite{Beny}, an approach that we believe is
subtly flawed.  A preserved measurement $\cM$ represents perfectly
preserved information \emph{only} if there is some circumstance under
which Alice would measure $\cM$ in order to answer a question.  If
$\cM$ is inherently noisy and error-laden, then for any question Alice
might ask, there is always some $\cM'$ that would yield a better
answer.  The fact that $\cM$ can be reproduced by Bob is irrelevant if
Alice would never choose to make that measurement.

\item \emph{The whole idea of a code is appropriate only in the
communication-theoretic paradigm, not the observational one.  If the
input to the channel is controlled by an oblivious system (e.g., a
distant star) rather than a cooperative sender, then the
receiver/observer cannot rely on preparation within the code.}  This
is correct -- and yet the framework works nonetheless.  If any
information is perfectly preserved by the channel, then there
\emph{must} be at least two input states that remain distinguishable
at the output.  Conversely, if the channel mixes up \emph{every} pair
of input states, then there is absolutely no question that Bob can
answer as well as Alice.

It is true that the semantic meaning of a ``code'' is inappropriate to
the observational paradigm, since an oblivious ``sender'' is unlikely
to cooperate by carefully preparing within a code.  Ultimately, this
is why we focus not on codes, but on the underlying IPS.  The
existence of a preserved code is merely a symptom of the underlying
structure; if a code exists, then there is potentially an entire
equivalence class of codes.  This is especially true in the case of
unconditionally preserved information (the only kind relevant to
observation), where the recovery map $\hat\cE$ [recall
Eq. (\ref{uncondrec})] does not depend on any prior information about
the code ({\em e.g.}, a subspace projector $P$).  An unconditionally
preserved IPS is isometric to a subalgebra that spans the system's
entire Hilbert space (rather than a subspace $\cP$).  Every observable
in this algebra can be observed faithfully by the observer at the
channel's output.  Thus, in this situation, the code framework is
ancillary to the real question -- but it works nonetheless.

\end{itemize}

%******************************************************************
%******************************************************************
\section{Proofs}
\label{sec:Proofs}

In this section, we present complete proofs of the technical results
stated in the main text.

\subsection{Preserved information is correctable} 
\label{sub:preserved}

%Next, let us turn to preserved codes of $\cE$. The crucial step towards obtaining all preserved IPSs of $\cE$ is the following equivalence between preserved and correctable codes:
%************************ Theorem/Lemma ***************************
\begin{theorem}\label{thm:PresCorr}
\THMPRESCORR
\end{theorem}
\begin{proof}
The ``only if'' direction is straightforward.  For any
$\rho,\sigma\in\cC$, any $p\in [0,1]$, define the weighted difference
$\Delta=p\rho-(1-p)\sigma$.  If $\cC$ is correctable, then there
exists a CPTP $\cR$ such that, for every such $\Delta$,
$\Vert\Delta\Vert_1=\Vert(\cR\circ\cE)(\Delta)\Vert_1$.  The trace
norm is contractive under CPTP maps \cite{Contractivity}, so
$$\Vert(\cR\circ\cE)(\Delta)\Vert_1\leq\Vert\cE(\Delta)
\Vert_1\leq\Vert\Delta\Vert_1.$$
\noindent 
Combining these two expressions yields
$\Vert\cE(\Delta)\Vert_1=\Vert\Delta\Vert_1$, which means that $\cC$
is preserved by $\cE$.

To prove that preservation implies correctability, we give an explicit
correction operation.  This operation is known as the \emph{transpose
channel} \cite{Barnum00}, defined as
$$\cEP=\Pi\circ\cE^\dagger\circ\cN,$$ 
\noindent 
where $\cP$ is the joint support of all $\rho\in\cC$, $\Pi$ is the
projection onto $\cP$, $P$ is the projector onto $\cP$, $\cE^\dagger$
is the adjoint map of $\cE$, and $\cN$ is a normalization map given
below.  If the operator sum representation of $\cE$ is
\begin{equation*}
\cE(\rho) = \sum_i{E_i\rho E_i^\dagger},
\end{equation*}
then the OSRs for these maps are:
\begin{eqnarray*}
\Pi(\rho) &=& P\rho P, \\ 
\cE^\dagger(\rho) &=& \sum_i{E_i^\dagger \rho E_i}, \\ 
\cN(\rho) &=& \cE(P)^{-\frac12}\rho\cE(P)^{-\frac12}, \\
\cEP(\rho) &=& \sum_i{\left(P E_i^\dagger
\cE(P)^{-\frac12}\right)\rho\left(\cE(P)^{-\frac12} E_i P\right)}.
\end{eqnarray*}
Note that the inverse in $\cE(P)^{-\frac12}$ is taken on the support
of $\cE(P)$.  It is simple to verify that $\cEP$ is a trace-preserving
CP map.

To prove that $\cEP$ corrects the code $\cC$, we need a couple of
technical lemmas.  The first makes rigorous the notion of a channel's
action on a subspace:

\begin{lemma1}
\label{lem1:Subsp}
Let $\cE:\cB(\cH)\to\cB(\cH')$ be a CP map, and $X_0$ be a positive
semidefinite operator on $\cH$.  If $X$ is an operator on the support
of $X_0$, then $\cE(X)$ is an operator on the support of $\cE(X_0)$.
\end{lemma1}
\begin{proof}
Both $X_0$ and $X$ are diagonalizable, so $X_0$ has a smallest
eigenvalue, and $X$ has a largest eigenvalue.  Thus for some
$\epsilon>0$, $X_0>\epsilon X$, which means that $X_0-\epsilon X>0$.
Since $\cE$ is CP, $\cE(X_0-\epsilon X) \geq 0$. Because it is linear,
$\cE(X_0) \geq \epsilon\cE(X)$.  This implies that $X$ is supported on
the support of $X_0$.\end{proof}

%\noindent\textit{Proof.} $Define $\cP$ as the support of $X_0$ and $\cP'$ as the support of $\cE(X_0)$, and let the operator sum decomposition of $\cE$ be $\cE(X)=\sum_i{E_i X E_i^\dagger}$.  Suppose there exists $X\in\cB(\cP)$ such that $\cE(X)\notin\cB(\cP')$.  Then the span of the $\{E_i\}$ must include some operator $E\propto|\phi\rangle\langle\psi|$ such that $|\psi\rangle\in\cP$, but $|\phi\rangle\notin\cP'$. Therefore, $\cE(|\psi\rangle\langle\psi|)\notin\cB(\cP')$, for it has support on $|\phi\rangle\langle\phi|$. However, $X_0$ has support on $|\psi\rangle$ since it is full-rank on $\cP$, so $\cE(X_0)\notin\cB(\cP')$, which contradicts the definition of $\cP'$. $\square$\\

Now, recall that discriminating between two code states involves a
binary (Helstrom) measurement that projects onto one of two orthogonal
subspaces.  Our second lemma states that if a channel $\cE$ preserves
a code $\cC$, it also preserves the orthogonality of these subspaces.

\begin{lemma1}
\label{lem1:Disjoint}
Let $\cE$ be a CP map, $\rho$ and $\sigma$ be states in a code $\cC$
that is preserved by $\cE$, and $p\in[0,1]$.  Let us write $\Delta =
p\rho-(1-p)\sigma$ in terms of its positive and negative parts, as
$\Delta = \Dp-\Dm$, where $\Dpm$ are positive operators with disjoint
supports.  Then $\cE(\Dp)$ and $\cE(\Dm)$ have disjoint supports.
\end{lemma1}
\begin{proof}
The triangle inequality for the trace norm, together with the fact
that $\cE$ is TP, gives
\begin{align}
\Vert\cE(\Delta)\Vert_1 &=\Vert\cE(\Delta_+)-\cE(\Delta_-)\Vert_1\notag\\
&\leq\Vert\cE(\Delta_+)\Vert_1+\Vert\cE(\Delta_-)\Vert_1\notag\\
& =\tr(\Delta_+)+\tr(\Delta_-).
\label{contr}
\end{align}
Because $\cC$ is preserved,
$\Vert\cE(\Delta)\Vert_1=\Vert\Delta\Vert_1=\tr(\Dp)+\tr(\Dm)$.  This
implies equality throughout Eq. \eqref{contr}, that is,
$\Vert\cE(\Dp)-\cE(\Dm)\Vert_1=\Vert\cE(\Dp)\Vert_1+\Vert\cE(\Dm)\Vert_1$.
This is possible if and only if $\cE(\Dp)$ and $\cE(\Dm)$ have
disjoint supports.
\end{proof}

Armed with these results, we wish to prove that $\cC$ is noiseless for
$\cEP\circ\cE$.  To do so, we will show that for \emph{every} Helstrom
measurement $\{\cP_+,\cP_-\}$ that distinguishes between two states in
$\cC$, the subspaces $\cP_{\pm}$ are invariant under $\cE$.  First, we
prove this for the special case where the measurement forms a
partition of $\cP$ (that is, $\Delta$ is full-rank).

\begin{lemma1}
\label{lem1:InvPartition}
Define $\cE$ and $\Delta$ as in Lemma \ref{lem1:Disjoint}.  Define
$\cP_\pm\equiv \supp(\Dpm)$ and $P_\pm$ as the projector onto
$\cP_\pm$.  Then, if $\Delta$ is full-rank on $\cP$, then $\cP_+$ and
$\cP_-$ are invariant subspaces under $\RP\circ\cE$.
%, i.e., for any $X\in\cB(\cP_\pm)$, $(\RP\circ\cE)[X]\in\cB(\cP_\pm)$.
\end{lemma1}
\begin{proof}
$\RP$ is a composition of three CP maps, so $\RP\circ\cE$ can be
written as a composition of four maps:
$\RP\circ\cE=\Pi\circ\cE^\dagger\circ\cN\circ\cE$.  Let us define the
subspaces $\cQ_\pm\equiv \supp(\cE(\Dpm))$, and $Q_\pm$ as the
projectors onto $\cQ_\pm$.  We will prove the lemma by following the
subspaces $\cP_{\pm}$ through each of the four maps.

By Lemma \ref{lem1:Subsp}, $\cE$ maps every operator on $\cP_+$ to an
operator on $\cQ_+$, and every operator on $\cP_-$ to one on $\cQ_-$.
By Lemma \ref{lem1:Disjoint}, $\cQ_\pm$ are disjoint.  Thus, $\cE$
maps $\cP_\pm$ to disjoint subspaces $\cQ_\pm$.

Now we consider $\cN$.  $P_\pm$ and $\Delta_\pm$ have the same
support, so $\cE(P_\pm)$ is supported on $\cQ_\pm$.  Thus, $\cE(P_+)$
and $\cE(P_-)$ have disjoint supports, and because $P = P_+ + P_-$,
$$\cE(P)^{-1/2}=\cE(P_+)^{-1/2}+\cE(P_-)^{-1/2},$$
and so $\cN$ maps $\cQ_+\to\cQ_+$ and $\cQ_-\to\cQ_-$.

Now we consider $\cE^\dagger$.  Using the cyclic property of the
trace, $\tr(Q_\pm\cE(P_\mp))=0$ implies
$\tr(P_\mp\cE^\dagger(Q_\pm))=0$.  By Lemma \ref{lem1:Subsp},
$\cE^\dagger$ does not map $\cQ_\pm$ into $\cP_\mp$, which means that
$\cE^\dagger$ maps $\cQ_\pm$ to $\cP_\pm$.

Thus, $\cE^\dagger\circ\cN\circ\cE$ maps
$\cP_\pm\to\cQ_\pm\to\cQ_\pm\to\cP_\pm$.  The final projection $\Pi$
has no effect on any operator in $\cP$, so $\RP\circ\cE$ maps
$\cP_\pm\to\cP_\pm$.
\end{proof}

Lemma \ref{lem1:InvPartition} is the core of the proof for Theorem
\ref{thm:PresCorr}.  To complete the proof, we need to extend it to
cases where $\Delta$ is not full rank, and therefore $\{\cP_+,\cP_-\}$
do not form a partition of $\cP$.

\begin{lemma1}
\label{lem1:InvSubsp}
Lemma \ref{lem1:InvPartition} holds even if $\Delta$ is not full-rank
on $\cP$.
\end{lemma1}
\begin{proof}
%Suppose that $\Delta$ is not full-rank on $\cP$.  There are two ways that this can happen:
%\begin{enumerate}
%\item $\rho+\sigma$ is full-rank, so the zeros in the spectrum of $\Delta$ appear only for very specific values of $p$
%\item $\rho$ and $\sigma$ \emph{both} vanish on some subspace $\cP'$ of $\cP$.
%\end{enumerate}
%In case (1), consider a sequence of probabilities $\{p'_n\}$ that converges to $p$.  Each $p'_n$ defines a $\Delta'_n$ which \emph{is} full rank, so by Lemma \ref{lem1:InvPartition}, the corresponding subspaces.  The corresponding subspaces $\Delta'_\pm$ converge to $\Delta_\pm$, and Eq. \ref{eq:TrNorm} holds in the limit as $p'\to p$.
There exists a full-rank (on $\cP$) state $\rho_0\in\cC$.  This
follows because $\cP$ is the support of $\cC$, and $\cC$ is convex.
For any $\epsilon\in(0\ldots1)$, $(1-\epsilon)\rho + \epsilon\rho_0$
is full rank.  So we consider, in place of $\rho$, a sequence of full
rank states $\{\rho'_n\}$, where $\rho'_n = (1-\epsilon_n)\rho +
\epsilon_n\rho_0$, and $\{\epsilon_n\}$ converges to 0.  Lemma
\ref{lem1:InvPartition}, applied to the sequence of full-rank weighted
differences $\Delta'^{(n)} = p\rho'_n - (1-p)\sigma$, implies that the
corresponding partitions $\{\cP'^{(n)}_+,\cP'^{(n)}_+\}$ are invariant
subspaces.  As $n\to\infty$, $\Delta'^{(n)}_-$ converges to
$\Delta_-$, and $\cP'^{(n)}_-$ converges to $\cP_-$, while
$\cP'^{(n)}_+$ converges to the orthogonal complement of $\cP_-$ in
$\cP$.  Thus $\cP_-$ is invariant under $\cEP\circ\cE$.  The same
argument, but with $\sigma$ replaced by $\sigma'_n =
(1-\epsilon_n)\sigma + \epsilon_n\rho_0$, shows that $\cP_+$ is
invariant under $\cEP\circ\cE$.
\end{proof}

Armed with Lemmas \ref{lem1:InvPartition} and \ref{lem1:InvSubsp}, it
is now easy to prove that $\cC$ is noiseless for $\cEP\circ\cE$.
Consider an arbitrary convex combination of powers of $\cE$,
$$\cF\equiv\sum_np_n(\RP\circ\cE)^n,$$ 
\noindent 
where $\{p_n\}$ is a probability distribution over non-negative
integers.  Let $\Delta$ be a weighted difference of code states.  By
Lemmas \ref{lem1:InvPartition}-\ref{lem1:InvSubsp}, the supports of
$\Delta_+$ and $\Delta_-$ are invariant and disjoint subspaces.  Since
$\cF$ is trace-preserving,
\begin{eqnarray}
\Vert\cF(\Delta)\Vert_1&=&\tr(\cF(\Delta_+))+\tr(\cF(\Delta_-))
\label{eq:TrNorm}\\
&=&\tr(\Delta_+)+\tr(\Delta_-)=\Vert\Delta\Vert_1. \nonumber
\end{eqnarray}
This condition -- satisfied for all $\Delta$ -- is sufficient for
$\cC$ to be noiseless.
\end{proof}

\subsection{The structure of noiseless codes}

%************************ Theorem/Lemma ***************************

\begin{lemma}\label{lem:NSFixed}
\LEMNSFIXED
%\LEMNSFIXEDEXTRA
\end{lemma}

\begin{proof}
Consider the CPTP map
$$\cE_\infty=\lim_{N\rightarrow\infty}\frac{1}{N+1}\sum_{n=0}^N\cE^n.$$
The limit is well-defined for any map on a finite-dimensional Hilbert
space.  Note that $\cE\circ\cE_\infty = \cE_\infty$, so
$\cE[\cE_\infty(\rho)] =\cE_\infty(\rho)$ for any $\rho\in\cC$.  That
is, $\cE_\infty$ projects onto the fixed points of $\cE$.  Now, if
$\cC$ is noiseless for $\cE$, then it is preserved by any convex
combination of powers of $\cE$, and hence by $\cE_\infty$.  Since
$\cC$ is preserved by $\cE_\infty$, $\cC$ is isometric to
$\cE_\infty(\cC)$ (see Definition \ref{def:pres}).  As noted above,
$\cE_\infty(\cC)$ consists entirely of fixed states, so $\cC$ is
isometric to a set of fixed states.
\end{proof}

\begin{corollary}\label{cor:NSMax}
\CORNSMAX
\end{corollary}
\begin{proof}
Let $\cC$ be a noiseless code for $\cE$.  By Lemma \ref{lem:NSFixed},
$\cC$ is isometric to a subset of the fixed states.  The fixed states
themselves form a noiseless code $\cC_{\mathrm{max}}$.  If $\cC$ is
isometric to a proper subset of the fixed states, then $\cC$ is
strictly smaller than $\cC_{\mathrm{max}}$, and is therefore not
maximum.
\end{proof}

A similar result for preserved codes follows from the fact that they
can be made noiseless (Theorem \ref{thm:PresCorr}).

%************************ Theorem/Lemma ***************************
\begin{theorem}\label{thm:PresFixed}
\THMPRESFIXED
\end{theorem}
\begin{proof}
This follows from combining Lemma \ref{lem:NSFixed} with Theorem
\ref{thm:PresCorr} and Definition \ref{def:correctable}.
\end{proof}

These results tell us that maximum preserved codes have the same
structure as fixed-state sets -- but not what that structure is.  The
following theorem fills that gap, defining the structure of an
arbitrary CPTP map's fixed points.  It also characterizes the fixed
points of the adjoint map $\cE^\dagger$ (defined so that if $\cE(\rho)
= \sum_i{E_i\rho E_i^\dagger}$, then $\cE^\dagger(\rho) =
\sum_i{E_i^\dagger \rho E_i}$).  This extra result is useful in
Section \ref{sec:App}, in the algorithm for finding noiseless codes of
$\cE$.

%************************ Theorem/Lemma ***************************
\begin{theorem}\label{thm:FixedPtThm}
\THMFIXEDPTTHM
\end{theorem}

\begin{proof} 
First, we will prove that $\cP_0$ is an invariant subspace under
$\cE$, using the following lemma.
\begin{lemma1}
$\FixE$ contains a positive, full-rank (on $\cP_0$) operator.%, i.e.,
there exists $\rho_0\in\FixE$, such that
$\langle\psi|\rho_0|\psi\rangle>0$ for all pure states
$|\psi\rangle\in\cP_0$.
\end{lemma1}
\begin{proof}
Let $\rho_0\equiv\cE_\infty({\Id})$, where ${\Id}$ is the
identity on $\cH$. Since $\cE_\infty$ is CP and projects onto fixed
points of $\cE$, $\rho_0$ must be a non-negative fixed point of $\cE$,
and hence is in $\FixE$. Let $\cQ\subseteq\cP_0$ be the support of
$\rho_0$. We want to show that $\cQ=\cP_0$. Suppose $\cQ$ is a proper
subspace of $\cP_0$. Then, there exists $|\psi\rangle$ in
$\cP_0\backslash\cQ$ such that $\langle\psi|\rho_0|\psi\rangle=0$, but
there exists $X\in\FixE$ such that $\langle\psi|X|\psi\rangle\neq 0$.
%(for example, taking $X$ to be a full-rank fixed operator on $\cP_0$ works).
Let $Y$ be one of the four possible Hermitian operators:
$\pm(X+X^\dagger)$, $\pm i(X-X^\dagger)$, chosen so that
$\langle\psi|Y|\psi\rangle<0$ (this must be true for at least one of
the four possibilities). Since $X^\dagger$, $-X$ and $iX$ are all in
$\FixE$ if $X\in\FixE$, $Y$ is also in $\FixE$, so
$\cE_\infty(Y)=Y$. Now consider the operator $\rho={\Id}+\delta
Y$, where $\delta>0$ is chosen small enough so that $\rho$ is
non-negative. Then, $\cE_\infty(\rho)=\rho_0+\delta Y$. However,
$\langle\psi|\rho|\psi\rangle<0$, which contradicts the CP property of
$\cE_\infty$. Therefore, $\cQ=\cP_0$, and $\rho_0$ is the desired
positive, full-rank fixed operator.
\end{proof}

Applying Lemma \ref{lem1:Subsp} to $\rho_0$ implies that $\cP_0$ is an
invariant subspace under $\cE$, which proves part (i) of the theorem.

Now, to prove part (ii), we consider $\EP\equiv
\Pi_0\circ\cE\circ\Pi_0$, the restriction of $\cE$ to $\cP_0$.  Its
Kraus operators are $\{K_i\}=\{P_0E_iP_0\}$, where $P_0$ is the
projector onto $\cP_0$. Since $\cP_0$ is an invariant subspace,
$E_iP_0=P_0E_iP_0$ $\forall i$, which means that $\EP$ is TP, {\em
i.e.} $\sum_i K_i^\dagger K_i=P_0$. Furthermore, since all of $\cE$'s
fixed points are supported on $\cP_0$, $\EP$ has the same fixed points
as $\cE$.

We can now show that $\EP^\dagger$'s fixed points must commute with
its Kraus operators.
\begin{lemma1}
\label{lem1:Commute}
For any $X\in\mathcal{B}(\cP_0)$, $\EP^\dagger(X)=X$ if and only if
$[X,K_i]=0$ for all $i$.
\end{lemma1}
\begin{proof}
If $[X,K_i]=0\:\forall i$, then 
$$\EP^\dagger(X)=\sum_i K_i^\dagger X K_i=\Big(\sum_i K_i^\dagger
K_i\Big)X=P_0X=X.$$
\noindent 
Conversely, suppose $\EP^\dagger (X)=X$.  Consider the quantity
$$
\sum_i[X,K_i]^\dagger[X,K_i] = 
%&&\EP^\dagger(X^\dagger X)-X^\dagger \EP^\dagger(X)-(\EP^\dagger(X))^\dagger X+X^\dagger X \\
\EP^\dagger(X^\dagger X)-X^\dagger X, $$ 
\noindent 
after some algebra. By construction, this is non-negative.  Now,
observe that
\begin{eqnarray*}
\Tr\{\rho_0[\EP^\dagger(X^\dagger X)-X^\dagger
X]\}&=&\Tr\{\EP(\rho_0)X^\dagger X\}\\&&-\Tr\{\rho_0X^\dagger X\} \\
&=&0,
\end{eqnarray*}
since $\rho_0$ is fixed under $\cE$ (and hence $\EP$).  Because
$\rho_0$ is full-rank and positive, for any positive operator
$Y\in\cB(\cP_0)$, $\Tr(\rho_0Y)=0\Leftrightarrow Y=0$. Therefore,
$\EP^\dagger(X^\dagger X)-X^\dagger X=0$, and
$\sum_i[X,K_i]^\dagger[X,K_i]=0$. Since every term in the sum is
non-negative, we conclude that $[X,K_i]=0~\forall i$.  (Note: This
proof is adapted from a result in \cite{Lin99a}.)
\end{proof}

Lemma \ref{lem1:Commute} tells us that the fixed points of
$\EP^\dagger$ are precisely the commutant in $\cB(\cP_0)$ of the Kraus
operators $\{K_i\}$.  Commutants are closed under addition and
multiplication, and the fixed points of $\EP^\dagger$ are closed under
Hermitian conjugation.  Therefore, the fixed points of $\EP^\dagger$
form a matrix algebra, which completes the proof of part (ii) of the
theorem.

Let us denote this matrix algebra $\cA$.  The structure theorem for
matrix algebras (see Eq. (\ref{decomp}) and Ref. \onlinecite{Dav96a}) states
that, in some basis, we can write $\cA$ as
\begin{equation}
\label{decompA}
\cA\cong\bigoplus_k(\cM_{A_k}\otimes{\Id}_{B_k}),
\end{equation}
which induces a natural Hilbert space decomposition:
%% LV: Note change in notation A_k and B_k taken as Hilbert spaces!....
\begin{equation}
\label{eq:HDecomp}
\cH=\cP_0\oplus\overline{\cP_0}=\left[\bigoplus_k({A_k}
\otimes{B_k})\right]\oplus\overline{\cP_0}.
%\cH=\cP_0\oplus\overline{\cP_0}=\left[\bigoplus_k(\cH_{A_k}
%\otimes\cH_{B_k})\right]\oplus\overline{\cP_0}.
\end{equation}

In this basis, we can say something about the Kraus operators of
$\cE$.
\begin{lemma1}
\label{lem1:TriangularKrausForm}
Given a CPTP map $\cE$ on $\cB(\cH)$, let $\cP_0$ be the support of
its fixed points, and $\cA$ the algebra fixed by $\EP^\dagger$ (as in
Theorem \ref{thm:FixedPtThm}).  In the decomposition of $\cH$ induced
by $\cA$ (Eq. \ref{eq:HDecomp}), the Kraus operators of $\cE$ have the
form:
\begin{equation}
\label{eq:TriangularKrausForm}
E_i=\left(\begin{array}{cc}\bigoplus_k({\Id}\otimes K_{i,k})
&D_i\\0&C_i\end{array}\right),
\end{equation}
for some operators $K_{i,k}\in\cB({B_k})$,
$C_i\in\cB(\overline{\cP_0})$ and $D_i\in\cB(\overline{\cP_0},\cP_0)$.
\end{lemma1}
\begin{proof}
The $E_i$ operators can always be written in the $2\times 2$ block
form given above. Since $C_i$ and $D_i$ are arbitrary, we need only
show that the upper left block is of the given form, and that the
lower left block must vanish.  The upper left block of each $E_i$ is a
Kraus operator $K_i$ of $\EP$.  These are the Hermitian conjugates of
the Kraus operators for $\EP^\dagger$, which (by Lemma
\ref{lem1:Commute}) commute with $\cA$.  Therefore, they must be of
the form
$$K_i=\bigoplus_k({\Id}_{A_k}\otimes K_{i,B_k}),$$ 
\noindent 
which is the desired form for the upper left block.  Finally, we
observe that the lower left block maps operators on $\cP_0$ to
operators on $\overline{\cP_0}$.  Since $\cP_0$ is an invariant
subspace, this block must vanish.
\end{proof}

In light of the above, $\EP$ acts trivially on each of the
``noiseless'' factors ${A_k}$ factors, but does something nontrivial
on each of the ``noisy'' factors ${B_k}$ factors.  Furthermore, $\cE$
acts identically to $\EP$ on the $\cP_0$ subspace, but may do anything
at all to its complement (including mapping states on
$\overline{\cP_0}$ onto $\cP_0$).

The next step of the proof is to show that $\FixE$ is a distortion of
$\cA$.  Recall that $\cE$ and $\EP$ have the same fixed points, so we
need only characterize the fixed points of $\EP$.  We will do so by
constructing a vector space of fixed operators, then showing that this
exhausts the fixed points of $\EP$.

\begin{lemma1}
\label{lem1:FPDistortion}
Following the notation in Theorem \ref{thm:FixedPtThm}, let
$$\cA = \bigoplus_k{\cM_{A_k}\otimes{\Id}_{B_k}}$$ be the
algebra fixed by $\EP^\dagger$.  Then there exist positive
semidefinite operators $\tau_k \in \cB(\cH_{B_k})$ such that the
following distortion of $\cA$,
$$\tilde\cA = \bigoplus_k{\cM_{A_k}\otimes\{\tau_{B_k}\}},$$
consists entirely of operators that are fixed by $\cE$.
\end{lemma1}
\begin{proof}
Let $X = \sum_k{X_{A_k}\otimes\tau_{B_k}}$ be an element of
$\tilde\cA$.  By Lemma \ref{lem1:TriangularKrausForm},
\begin{eqnarray*}
\cE(X) &=& \sum_i{K_i X K_i^\dagger} \\ &=&
\sum_k{X_{A_k}\otimes\Big(\sum_i{K_{i,k}\tau_k
K_{i,k}^\dagger}\Big)} \\ 
&=& \sum_k{X_{A_k}\otimes \cE_{B_k}(\tau_k)}
\end{eqnarray*}
where for each $k$, $\cE_{B_k}: \cB(\cH_{B_k})\to\cB(\cH_{B_k})$ is a
CPTP map with Kraus operators $\{K_{i,k}\}$.  Schauder's fixed point
theorem \cite{GranasBook03} states that every CPTP map has at least
one fixed point.  If we let $\tau_{k}$ be a fixed point of
$\cE_{B_k}$, then $\cE(X) = X$.
\end{proof}

Now we need to show that $\tilde\cA$ contains \emph{all} the fixed
points of $\cE$.
\begin{lemma1}
Following the notation in Theorem \ref{thm:FixedPtThm}, let $\cA$ be
defined as in Lemma \ref{lem1:FPDistortion}.  Then every fixed point
of $\cE$ is in $\tilde\cA$.
\end{lemma1}
\begin{proof}
$\tilde\cA$ is closed under linear combination, so it is a vector
subspace of $\cB(\cP_0)$.  Its dimension is easily calculated:
%% LV: Check change!
$$\mathrm{dim}(\tilde\cA) =
\mathrm{dim}(\cA)\sum_k{\mathrm{dim}({A_k})^2}.$$
\noindent 
Let us view $\EP$ and $\EP^\dagger$ as matrices ($L$ and $L^\dagger$,
respectively) that act on vectors in $\cB(\cP_0)$. Since each element
of $\cA$ is fixed by $\cE$, and is therefore an eigenvector of $\EP$
with eigenvalue $+1$, $\EP$ has a +1 eigenspace of dimension at least
$\mathrm{dim(\cA)}$.  Furthermore, if $\cE$ had another fixed point
outside of $\cA$, then $\EP$'s +1 eigenspace would be strictly larger
than that.

Let $\{O_i\}$ be an orthonormal basis (in the Hilbert-Schmidt inner
product) for $\cB(\cP_0)$.  $L$ has matrix elements
$L_{ij}=\tr\{O_i^\dagger\EP(O_j)\}$, and $L^\dagger$ is its Hermitian
conjugate.  The eigenvalues of a matrix and its Hermitian conjugate
are complex conjugates of each other.  Thus, the dimensions of the
+1-eigenspaces of $\cL_{\EP}$ and $\cL_{\EP}^\dagger$ are equal, and
$\FixE$ and $\mathrm{Fix}(\EP^\dagger) = \cA$ have the same dimension.
So $\cE$ has no fixed points outside of $\tilde\cA$.
\end{proof}
These two lemmas prove that $\FixE=\tilde\cA$ is a distortion of $\cA$.

Finally, let us consider the fixed points of $\cE^\dagger$.  We begin
by showing that they are in 1:1 correspondence with the fixed points
of $\EP^\dagger$, by showing that $P_0\FixEdag P_0=\cA$.  The first
step is relatively straightforward.
\begin{lemma1}\label{lem1:OneToOne}
Following the notation in Theorem \ref{thm:FixedPtThm}, $P_0\FixEdag
P_0\subseteq\cA$.
\end{lemma1}
\begin{proof}
The Kraus operators of $\cE^\dagger$ are (by
Eq.\ref{eq:TriangularKrausForm} in Lemma
\ref{lem1:TriangularKrausForm})
$$ E_i^\dagger=\left(\begin{array}{cc}\bigoplus_k({\Id}\otimes
K^\dagger_{i,k}) & 0 \\ D^\dagger_i &
C^\dagger_i\end{array}\right). $$ Let $X$ be an element of $\FixEdag$.
By writing $X$ in block-diagonal form with respect to the
decomposition $\cH = \cP_0\oplus\overline{\cP_0}$, and noting that
$\cE^\dagger(X) = \sum_i{E_i^\dagger X E_i}$, it is straightforward to
show that
$$\EP^\dagger(P_0 X P_0) = P_0 \cE^\dagger(X) P_0,$$ and since
$\cE^\dagger(X) = X$, we conclude that $P_0XP_0$ is a fixed point of
$\EP^\dagger$, and therefore is an element of $\cA$.  So $P_0\FixEdag
P_0\subseteq\cA$.
\end{proof}

Now we need to show that $\cA\subseteq P_0\FixEdag P_0$.  This is a
bit more difficult, and requires a technical lemma.  Let us partition
the Hilbert-Schmidt space into subspaces as follows:
\begin{eqnarray*}
\cK &\equiv& \cB(\cH), \\
\cK_0 &\equiv& \cB(\cP_0), \\
\overline{\cK}_0 &\equiv& \cK / \cK_0 .
\end{eqnarray*}
We can write the matrix representing $\cE$ in block form as
\begin{equation}\label{Eblock}
L=\left( \begin{array}{cc} L_{\EP} & L_\cG\\0& L_\cF \end{array} \right).
\end{equation}
Here, $L$ corresponds to the map $\cE$, which acts on vectors in
$\cK$.  $L_{\EP}$ corresponds to the map $\EP$ and maps $\cK_0$ back
into itself.  $L_\cF$ maps $\overline{\cK}_0$ back into itself, while
$L_\cG$ maps $\overline{\cK}_0$ to $\cK_0$.  Because $\cP_0$ is an
invariant subspace, $L$ does not map $\cK_0$ to $\overline{\cK}_0$.
The matrix for $\cE^\dagger$ is the Hermitian conjugate
$L_\cE^\dagger$.

\begin{lemma1}
\label{FLem}
$L_\cF$ has no fixed points.
\end{lemma1}
\begin{proof}
Suppose there exists $X\in\overline{\cK}_0$ such that $L_\cF(X)=X$.
Define $Y=L_\cG(X)$.  Then
$$L_\cE\binom{0}{X}=\binom{Y}{X},$$ and the action of $\cE^n$ on the
operator corresponding to $\binom{0}{X}$ is given by
$$(L_\cE)^n\binom{0}{X}=\left(
\begin{array}{c}\sum_{m=0}^{n-1}L_{\EP}^m(Y)\\ X\end{array}\right).$$
If $Y$ is orthogonal to the subspace $\FixE$, then as $n\to\infty$,
the sum converges to
$$\lim_{n\to\infty}(L_\cE)^n\binom{0}{X}=\left(
\begin{array}{c}({\Id}-\EP)^{-1}(Y)\\ X\end{array}\right).$$
This is a fixed point of $\cE$ not contained in $\FixE$, which
contradicts the definition of $\FixE$.  On the other hand, if $Y$ is
\emph{not} orthgonal to $\FixE$, then the sum diverges as
$n\rightarrow\infty$. This implies that $\cE$ is non-contractive,
which violates complete positivity \cite{Contractivity}.  So, either
way, we have a contradiction.
\end{proof}

Using Lemma \ref{FLem}, we can show that every fixed point of
$\EP^\dagger$ has an extension to a fixed point of $\cE^\dagger$:
\begin{lemma1}\label{lem1:Extension}
Let $X_0\in\cA$ be a fixed point of $\EP^\dagger$.  Then there exists
a fixed point $X\in\cB(\cH)$ of $\cE^\dagger$ such that $P_0XP_0=X_0$.
\end{lemma1}
\begin{proof}
Both $X_0$ and $X$ are vectors in the Hilbert-Schmidt space $\cK =
\cB(\cH)$.  Using the decomposition $\cK = \cK_0 \oplus
\overline{\cK_0}$, we can write $X_0$ in block form:
$$X_0 = \binom{X_0}{0}.$$
In this block form, we choose
$$X = \binom{X_0}{({\Id}_{\overline{\cK}_0}-L_\cF^\dagger)^{-1}L_\cG^\dagger X_{\cK_0}}.$$
Note that $L_\cF$ has no fixed points (by Lemma \ref{FLem}), so
${\Id}_{\overline{\cK}_0}-\cL_\cF^\dagger$ is invertible, which
means that $X$ is well-defined.  Furthermore, $P_0XP_0 = X_0$ by
construction.  To show that $X$ is a fixed point of $\cE^\dagger$, we
simply compute
\begin{eqnarray*}
L^\dagger(X) &=& \left( \begin{array}{cc} L^\dagger_{\EP} & 0 \\
 L^\dagger_\cG & L^\dagger_\cF \end{array} \right)
 \binom{X_0}{({\Id}_{\overline{\cK}_0}-L_\cF^\dagger)^{-1}L_\cG^\dagger X_0} \\ &=&
 \binom{ L^\dagger_{\EP}(X_0)}{\left({\Id} +
 L_{\cF}^\dagger\left({\Id} - L_{\cF}^\dagger\right)^{-1}
 \right)L^\dagger_{\cG}(X_0)} \\ &=& \binom{ X_0}{\left({\Id} -
 L_{\cF}^\dagger + L_{\cF}^\dagger\right)\left({\Id} -
 L_{\cF}^\dagger\right)^{-1}L^\dagger_{\cG}(X_0)} \\ &=&
 \binom{X_0}{({\Id}_{\overline{\cK}_0}-L_\cF^\dagger)^{-1}L_\cG^\dagger X_0} \\ &=& X.
\end{eqnarray*}
\end{proof}

Lemma \ref{lem1:Extension} implies that $\cA\subseteq P_0\FixEdag
P_0$.  Combining this with Lemma \ref{lem1:OneToOne}, we conclude that
$\cA=P_0\FixEdag P_0$, which completes the proof of Theorem
\ref{thm:FixedPtThm}.
\end{proof}

Now, we want to show that $\cE$'s noiseless codes have a rigid
structure dictated by the fixed points.

%************************ Theorem/Lemma ***************************

\begin{lemma}
\label{lem:NSStruct}
\LEMNSSTRUCT
\end{lemma}

\begin{proof} 
If $\cC$ has the given structure, then:
\begin{enumerate}
\item It is maximum, since it is isometric to the full set of fixed
states of $\cE$.
\item It is noiseless, because $\cE$ leaves the states on subsystem
$A_k$ intact, and every $\rho_k$ state has the same noise-full state
$\mu_k$.  So $\cE$ preserves all the weighted 1-norm distances between
code states.
\end{enumerate}

To show the converse, we must show that if $\cC$ is \emph{not} of this
form, then it is not maximum noiseless.  If $\cC$ is not of this form,
then either

\begin{enumerate}
\item It contains only a strict subset of the states given above; or,
\item It contains at least one state with correlations (off-diagonal
elements) between different $k$-sectors; or
\item It contains at least one state with correlations between
${A_k}$ and ${B_k}$; or 
\item It contains states that differ on ${B_k}$.
\end{enumerate}
If $\cC$ is a strict subset, then it is obviously not maximum.

The key to proving the converse is showing that the condition for
noiselessness (Definition \ref{def:noiseless}) forbids correlations
between the $k$-sectors as well as between ${A_k}$ and ${B_k}$.  The
proof relies both on convexity and on the code being maximum.  First,
recall the map $\cE_\infty$ from Lemma \ref{lem:NSFixed}, which
projects onto the fixed point set $\FixE$. Given the structure of
$\FixE$, the CPTP $\cE_\infty$ must act on states on $\cP_0$ as:
\begin{equation}
\label{Einfty}
\cE_\infty(\rho)=\bigoplus_k~\big(\tr_{B_k}\{P_k\rho P_k\}\otimes
\tau_{B_k}\big),
\end{equation}
where $\tau_{B_k}$ is the fixed state on ${B_k}$ from Theorem
\ref{thm:FixedPtThm}, and $P_k$ projects onto the $k$th sector. From
Lemma \ref{lem:NSFixed}, we know that for every fixed state of the
form $\rho_f\equiv\bigoplus_k (\sigma_{A_k}\otimes \tau_{B_k})$, there
exists exactly one code state $\rho\in\cC$ such that
$\cE_\infty(\rho)=\rho_f$. From Eq. \eqref{Einfty}, this demands
$\tr_{B_k}\{P_k\rho P_k\}=\sigma_{A_k}$ for all $k$.

Now, focus on the case with only two $k$-sectors, labeled 1 and
2. Consider two fixed states in these sectors with block-diagonal
form:
\begin{equation*}
\rho_{f1}=\left(\begin{array}{cc}\rho'_{f1}&0\\0&0\end{array}\right),
\quad
\rho_{f2}=\left(\begin{array}{cc}0&0\\0&\rho'_{f2}\end{array}\right).
\end{equation*}
The two code states that are isometric to the fixed points must
respectively be of the form
\begin{equation*}
\rho_1=\left(\begin{array}{cc}\rho'_1&0\\0&0\end{array}\right),\qquad
\rho_2=\left(\begin{array}{cc}0&0\\0&\rho'_2\end{array}\right).
\end{equation*}
By convexity of $\cC$, any convex combination of $\rho_1$ and $\rho_2$
must also be in $\cC$. This excludes from $\cC$ any state with
on-diagonals equal to this convex combination, but non-zero
off-diagonals, since the two different states will have the same image
(and hence indistinguishable) under $\cE_\infty$. Generalizing this to
any number of $k$-sectors, we find that any code state in $\cC$ must
be block-diagonal: $\rho=\bigoplus_k\rho'_k$.

Next, consider the state $\rho'_k$ for the $k$th sector. We need to
show that only product states of ${A_k}\otimes{B_k}$ are allowed. We
first consider a fixed state $\rho'_f$ on this sector of the form
$|\psi\rangle\langle \psi|_{A_k}\otimes \tau_{B_k}$. Since the state
on ${A_k}$ is pure, the corresponding code state whose image under
$\cE_\infty$ is $\rho'_f$ must also be pure on ${A_k}$. It is hence a
product state of the form $|\psi\rangle\langle\psi|_{A_k}\otimes
\mu_{B_k}$. Next, suppose $\rho'_f=\sigma_{A_k}\otimes \tau_{B_k}$,
where $\sigma_{A_k}$ is in general a mixed state writable as
$\sigma_{A_k}=\sum_\alpha
q_\alpha|\psi_\alpha\rangle\langle\psi_\alpha|_{A_k}$. Now, each state
$|\psi_\alpha\rangle\langle\psi_\alpha|_{A_k}\otimes \tau_{B_k,0}$ is
a fixed state, with corresponding code state
$\rho'_{k,\alpha}=|\psi_\alpha\rangle\langle\psi_\alpha|_{A_k}\otimes
\mu_{B_k,\alpha}$. By convexity, the state $\sum_\alpha
q_\alpha\rho'_{k,\alpha}$ is also in $\cC$ and maps to
$\rho_f=\sigma_{A_k}\otimes \tau_{B_k}$ under $\cE_\infty$. This
excludes from $\cC$ any other state with non-zero correlations between
${A_k}$ and ${B_k}$, but with the reduced state on ${A_k}$ equal to
$\sigma_{A_k}$. Furthermore, we must have that
$\mu_{B_k,\alpha}=\mu_{B_k}$ $\forall \alpha$ in order for the
(1-norm) distinguishability between the $\rho'_{k,\alpha}$'s to remain
unchanged under $\cE_\infty$. Therefore, $\rho'_k$ must be of the form
$\sigma_{A_k}\otimes \mu_{B_k}$ for some $\mu_{B_k}$.
\end{proof}

\noindent 
We knew already that noiseless codes are isometric to fixed states
(Lemma \ref{lem:NSFixed}) and that fixed states are isometric to
algebras (Theorem \ref{thm:FixedPtThm}).  Now we know explicitly what
these codes look like.  The isometry is very similar to the one
between the fixed states ($\FixE$) and the underlying algebra $\cA$: A
noiseless code is obtained from $\FixE$ just by changing the state of
the noise-full factors\footnote{Since $\cC$ only contains
\emph{states}, we are really restricting to the positive trace-1
operators in $\cM_k$ within $\FixE$ and $\cA$. This is what we mean by
``$\cC$ is isometric to a matrix algebra.''}.

Finally, it follows from this lemma that not only can we make
preserved codes noiseless, but we can also make them fixed.
\begin{corollary}\label{cor:FPCond}
For every maximum preserved code $\cC$, there exists a CPTP map $\cR$
such that $\cR\circ\cE(\rho)=\rho$ for all states $\rho\in\cC$.
\end{corollary}

\begin{proof}
From Theorem \ref{thm:PresCorr}, we know that every preserved code
$\cC$ is correctable, so there exists a recovery map $\cR_0$ such that
$\cC$ is noiseless for $\cR_0\circ\cE$, and $\cR_0\circ\cE$ is unital.
By Lemma \ref{lem:NSStruct}, $\cC$ contains states all of the form
$\rho = \sum_k(\rho_k\otimes \mu_k)$.  Now let $\cR = \cT\circ\cR_0$,
where $\cT$ does nothing to the $A_k$ subsystems, but replaces the
state of each $B_k$ subsystem with $\mu_k$.  (Constructing such a map
is simple, and it is manifestly CPTP).  Now, every $\rho\in\cC$ is a
fixed state of $\cR\circ\cE$.
\end{proof}

\subsection{Finding preserved IPS is hard}

%************************ Theorem/Lemma ***************************

\begin{lemma}\label{lem:FindingIPSisHard}
\LEMFINDINGIPSISHARD
\end{lemma}

\begin{proof}
The proof is straightforward, and proceeds in three steps.  First, we
review a known result connecting classical channels with graphs.
Second, we show that finding the largest code for a certain set of
classical channels is equivalent to \textbf{MAX-CLIQUE}.  Third, we
observe that the classical channels can be embedded in quantum
channels.

\begin{enumerate}
\item A classical channel $\cE_c$ maps a set of input symbols
$\{1\ldots N\}$ into mixtures of a set of output symbols $\{1\ldots
M\}$.  For each input symbol $n$, its \emph{image} $\cI(n)$ is the set
of output symbols to which $\cE$ maps it with nonzero probability.  A
set of input symbols $\cC=\{n_1\ldots n_k\}$ is a preserved zero-error
code for $\cE$ if and only if the images of all the $n_j$ are disjoint
-- i.e., it is possible to unambiguously identify which of the input
symbols was sent.  We can define the channel's \emph{adjacency graph}
$G$ (see Example \ref{ex:ClassicalFourState}) as follows: The vertices
are labeled by input symbols $\{1\ldots N\}$, and two vertices
$\{n,m\}$ are connected by an edge if and only if the images $\cI(n)$
and $\cI(m)$ are overlapping.  Now, a code $\cC$ is a subgraph of $G$,
and it is preserved if and only if no two of its vertices are
connected -- {\em i.e.}, if it is an \emph{independent set} of $G$.
The largest code is a \emph{maximum independent set} of $G$.  An
independent set for $G$ is a clique for its dual graph $G'$, and
finding the maximum clique for an arbitrary $G'$ is a well-known
NP-complete problem called \textbf{MAX-CLIQUE}.

\item We haven't yet shown that finding a classical channel's largest
code is NP-complete -- perhaps all channel's adjacency graphs are easy
instances of \textbf{MAX-CLIQUE}?  This turns out not to be the case;
any graph $H$ can be the adjacency graph of a classical channel.  Let
$H$ be a graph with vertices $\{1\ldots d\}$, and let $\cE$ be a
classical channel from $\{1\ldots d\}\to\{1\ldots d^2\}$, defined as
follows:

\begin{enumerate}
\item The $d$ input symbols are denoted $v\in\{1\ldots d\}$, and the
$d^2$ output symbols are denoted by ordered pairs $u\in\{1\ldots
d\}\times\{1\ldots d\}$.

\item For each input symbol $v\in\{1\ldots d\}$, $\cE$ maps $v$ (with
nonzero probability) to each of the $d$ output symbols
$\{(v,x):x=1\ldots d\}$.

\item For each input symbol $v$, $\cE$ maps each input symbol $v$ to
output symbol $(v',v)$ if and only if $H$ contains the edge $(v',v)$.
\end{enumerate}

Note that each output symbol $(a,b)$ can be produced by at most two
input symbols ($a$ and $b$).  So, if two input symbols $v$ and $v'$
are connected in $H$, then $\cE$ maps both of them to the output
symbol $(v',v)$, and so they are connected in the adjacency graph $G$.
But, if they are not connected in $H$, then they are not mapped to the
same output symbol, so they are not connected in $G$.  Ergo, $G=H$,
and any graph can be produced as the adjacency graph of a channel.

\item Finally, we need to show that for each such graph, we can
construct a quantum channel.  This is rather easy.  Let the input
space be $\cH_d$ and the output space be $\cH_{d^2}$.  Let
$\{\ket{1},\ldots,\ket{d}\}$ be a basis for $\cH_d$.  Then the $\cE$
we will consider acts as follows: First, it dephases in the given
basis ({\em i.e.}, measures it); and then it acts as the classical
channel above.
\end{enumerate}
\end{proof}

%*********************************************************************
\subsection{Unitarily noiseless codes}

The analysis of unitarily noiseless codes follows closely that of the
noiseless codes.  The rotating points of $\cE$ replace its fixed
points, with a CPTP map that projects onto their span playing the role
that $\cE_\infty$ does for noiseless codes.

%************************ Theorem/Lemma ***************************
\begin{lemma}\label{lem:UNRot}
\LEMUNROT
\end{lemma}

\begin{proof}
By Definition \ref{def:RotatingPoint}, a rotating point $X$ of $\cE$
is a linear combination of operators $X_k$ such that
$\cE(X_k)=e^{i\phi_k}X_k$.  Let $\RotE$ be the complex span of all
rotating points of $\cE$. It is convenient to move to the
Hilbert-Schmidt space, where $\RotE$ can be viewed as a subspace
spanned by the vectors corresponding to the rotating points. Clearly,
$\RotE$ is an invariant subspace under the linear map $\cE$, in the
sense that any vector in $\RotE$ gets mapped under $\cE$ to another
vector in $\RotE$. Let $\EQ$ denote $\cE$ restricted to $\RotE$. We
view $\cE$ and $\EQ$ as matrices acting on vectors in the
Hilbert-Schmidt space.

Even though $\cE$ may not be a diagonalizable matrix, we can still
write it in the Jordan normal form \cite{LinAlg}: There exists an
invertible matrix $S$ such that $\cE=SJS^{-1}$, where $J$ is the
matrix $J=\text{diag}[J_1,J_2,\ldots, J_K]$. Each $J_k$ is called a
\emph{Jordan block}, and it is zero except on the diagonal and
first-off-diagonal:
\begin{equation}\label{eq:Jk}
J_k=\left(
\begin{array}{cccc}
\lambda_k& 1& &\\
&\ddots&\ddots&\\
&&\lambda_k&1\\
&&&\lambda_k
\end{array}
\right).
\end{equation}
The Jordan form for $\cE$ is unique up to permutation of the Jordan
blocks. Note that any vector $|v\rangle$ is an eigenvector of $J$ if
and only if $S|v\rangle$ is an eigenvector of $\cE$.

\begin{lemma1}
\label{lem:UNRotJordan}
For any $k$, the support of $J_k$ contains exactly one unit
eigenvector of $\cE$. The corresponding eigenvalue is $\lambda_k$.
\end{lemma1}
\noindent {\it Proof.} Let $\{|v_\alpha^{(k)}\rangle\}_{\alpha=1}^m$
be the ordered basis for the support of $J_k$ in which $J_k$ takes the
form Eq. \eqref{eq:Jk}. Clearly,
$J_k|v_1^{(k)}\rangle=\lambda_k|v_1^{(k)}\rangle$, so
$S|v_1^{(k)}\rangle$ is an eigenvector of $\cE$ with eigenvalue
$\lambda_k$. To show that this is the only eigenvector in this Jordan
block, let $|v\rangle\equiv \sum_\alpha
\mu_\alpha|v_\alpha^{(k)}\rangle$ be a vector in the support of
$J_k$. From the form of $J_k$ in Eq. \eqref{eq:Jk}, it is easy to see
that the coefficients $\{\mu_\alpha\}$ satisfy the equation
$J_k|v\rangle=a|v\rangle$ for some constant $a$ only if
$\mu_{\alpha+1}=(a-\lambda_k)\mu_\alpha$ for $\alpha=1,\ldots, m-1$,
and $(a-\lambda_k)\mu_m=0$. The only non-trivial solution is
$a=\lambda_k$ and $\mu_1\neq 0,
\mu_{\alpha>1}=0$. \hfill$\blacksquare$

This lemma tells us that the rotating points of $\cE$ are
mutually orthogonal, unless there are degenerate eigenspaces of
rotating points. In that case, we can still pick an orthonormal basis
for each degenerate eigenspace (already done in the Jordan normal
form), and these bases, together with the non-degenerate rotating
points, form an orthonormal basis of rotating points for $\RotE$. We
denote this basis as $\{X_l\}$. $\EQ$ is diagonal in this basis, with
entries $e^{i\phi_l}(=\lambda_l)$. Note that, for any CPTP map $\cE$,
the following lemma from \cite{TerhalDiVincenzo} holds:
\begin{lemma1}\label{lem:UNRotEig}
Any eigenvalue $\lambda$ of $\cE$ must satisfy $|\lambda|\leq 1$.
\end{lemma1}
\noindent This, together with Lemma \ref{lem:UNRotJordan}, implies
that $|\lambda_k|\leq 1~\forall k$.

Next, consider powers of $\cE$. $\cE^n$ can be written using the
Jordan normal form as $SJ^nS^{-1}$ where $J^n=\text{diag}[J_1^n,
J_2^n,\ldots, J_K^n]$ with each $J_k^n$ being an upper-triangular
matrix:
\begin{equation}
\label{eq:Jkn}
J_k^n=\left(
\begin{array}{cccc}
\lambda_k^n&
\binom{n}{1}\lambda_k^{n-1}&\binom{n}{2}\lambda_k^{n-2}&\ldots\\
0&\lambda_k^n&\binom{n}{1}\lambda_k^{n-1}&\ldots\\
0&0&\lambda_k^n&\ldots\\ &&&\ddots
\end{array}
\right)
\end{equation}
Using the form of $J_k^n$ in Eq. \eqref{eq:Jkn}, we can show the
following fact about the rotating points of $\cE$:
\begin{lemma1}\label{lem:UNRot1D}
Any (non-degenerate) rotating point of $\cE$ must occur in a
1-dimensional Jordan block.
\end{lemma1}

\noindent {\it Proof.} (This proof follows ideas from
\cite{TerhalDiVincenzo} for the proof of Lemma \ref{lem:UNRotEig}.)
Suppose there exists a rotating point $X$ such that it belongs to some
$m\times m$ Jordan block $J_k$ with $m>1$.  Let
$\{X_\alpha^{(k)}\}_{\alpha=1}^m$ be an operator basis for the
operators in the support (as vectors) of $J_k$, with $X_1^{(k)}\equiv
X$. Consider the completely mixed state $\rho_{\Id}\equiv{\Id}/d$
($d$ is the dimension of the Hilbert
space). Let $\sigma$ be some operator in the span of
$\{X_\alpha^{(k)}\}_{\alpha=2}^m$ and consider the operator
$\rho\equiv \rho_{\Id}+\eta\sigma$ where $\eta$ is a positive
number chosen small enough so that $\rho$ is positive. Applying
$\cE^n$ to $\rho$ gives $\cE^n(\rho)=\cE^n(\rho_{\Id})+\eta\cE^n(\sigma)$.
Since $\cE$ is TP, $\cE^n(\rho_{\Id})$
remains finite. However, since $X$ is a rotating point, we know that
$|\lambda_k|=1$, and the entries of $J_k^n$ grows in amplitude as $n$
increases, and hence the entries of $\cE^n(\sigma)$ (viewed as a
vector) grow in amplitude. For large enough $n$ ($\eta$ fixed), there
will be a choice of $\sigma$ such that $\cE^n(\rho)$ is no longer
positive semidefinite. But this violates the assumption that $\cE$ is
a CPTP map. Hence, we must have that $m=1$. \hfill$\blacksquare$

\noindent Lemma \ref{lem:UNRot1D} tells us that any Jordan block $J_k$
with $m>1$ must have $|\lambda_k|<1$.

Now, let $\{Y_\beta\}$ be an operator basis for operators outside of
$\RotE$. $Y_\beta$'s are the operators occurring in Jordan blocks with
$|\lambda_k|<1$, and hence $\lim_{n\rightarrow\infty}\cE^n(Y_\beta)=0$
since Eq. \eqref{eq:Jkn} tells us that
$\lim_{n\rightarrow\infty}J_k^n=0$ if $|\lambda_k|<1$. We can use
$\{X_l\}\bigcup\{Y_\beta\}$ as an operator basis for $\cB(\cH)$, and
write any operator $A\in\cB(\cH)$ as $A=\sum_l a_lX_l+\sum_\beta
b_\beta Y_\beta$. Then,
\begin{align}
\label{eq:EinfA}
\lim_{n\rightarrow\infty}\cE^n(A)&=\lim_{n\rightarrow
\infty}\Big(\sum_la_l(\EQ)^n(X_l)+\sum_\beta b_\beta \cE^n
(Y_\beta)\Big)\notag\\ &=\sum_la_l\lim_{n\rightarrow
\infty}(\EQ)^n(X_l),
\end{align}
assuming the limit $\lim_{n\rightarrow \infty}(\EQ)^n(X_l)$ exists for
all $l$.

To work out what $\lim_{n\rightarrow\infty}(\EQ)^n(X_l)$ is, we need
the following lemma:

\begin{lemma1}
\label{lem2}
For every $\epsilon>0$, there exists some $\Ne\in\mathbb{N}$ such that
$\Vert (\EQ)^{\Ne}-{\Id}_R\Vert<\epsilon$, where ${\Id}_R$
is the identity operator on $\RotE$.
\end{lemma1}

\noindent
\textit{Proof.}  Recall that $\EQ$ is a diagonal matrix, with entries
$e^{i\phi_l}$, $l=1,\ldots, M$, where $M =
\mathrm{dim}(\RotE)$. Therefore, $(\EQ)^n$ is also diagonal, with
entries $e^{in\phi_l}$, and in particular $(\EQ)^0 = {\Id}$.
The set of \emph{all} such matrices forms an $n$-torus with a finite
volume $(2\pi)^M$.  Each $(\EQ)^n$ is surrounded by an
$\epsilon$-neighborhood $\cN_n$, containing all matrices $X$ on the
torus such that $\Vert (\EQ)^n-X\Vert<\epsilon$.  Each such
neighborhood has volume at least $\epsilon^M$, and so if we consider
the neighborhoods of $(\EQ)^n$ for $n=0\ldots (2\pi/\epsilon)^M$, then
at least one pair must overlap.  Denote the pair with overlapping
neighborhoods

If $\phi_l$'s are all rational multiples of $2\pi$, {\em i.e.},
$\phi_l=\frac{2\pi p_l}{q_l}$, $p_l,q_l\in\mathbb{N}$, then choosing
$\Ne$ to be the lowest common multiple of all $q_l$ works.

Otherwise, a more complicated analysis is required. To have $\Vert
(\EQ)^{\Ne}-{\Id}_R\Vert=\max_l\vert
\exp(i\Ne\phi_l)-1\vert=2\max_l\vert \sin(\Ne\phi_l/2)\vert<\epsilon$,
it suffices to demand $\Ne\phi_l (\text{mod } 2\pi)<\epsilon$ for all
$l$.  Consider the point $(n\phi_1(\textrm{mod }
2\pi),\ldots,n\phi_M(\textrm{mod } 2\pi))$, where we always take the
smallest non-negative value of $n\phi_l(\textrm{mod }2\pi)$. As $n$
increases from 0, this point traces out a trajectory on the surface of
an $M$-dimensional torus.  If there is at least one $\phi_l$ that is a
rational multiple of $2\pi$, this trajectory will eventually close
upon itself, and the path length of the trajectory is finite. If there
is no such $\phi_l$, the trajectory will cover the surface of the
torus, which has finite area (since it is
finite-dimensional). Consider hyperspheres of (Euclidean) diameter
$\epsilon$ centered at $(n\phi_1(\textrm{mod }
2\pi),\ldots,n\phi_M(\textrm{mod } 2\pi))$ for each
$n\in\mathbb{N}$. Because the trajectory either has finite length or
traverses a space of finite area, some of these hyperspheres will
eventually overlap, that is, there exists finite $r$ and $s>r$ such
that the hyperspheres centered at points with $n=r$ and $n=s$
overlap. The distance between the centers of the overlapping
hyperspheres is $\sqrt{\sum_l [(s-r)\phi_l(\textrm{mod
}2\pi)]^2}<\epsilon$, which implies that $(s-r)\phi_l(\textrm{mod
}2\pi)<\epsilon$ for all $l$. Therefore, we can choose
$\Ne=s-r$. \hfill$\blacksquare$

We can view the limit $\lim_{n\rightarrow\infty}(\EQ)^n$ equivalently
as the limit $\lim_{n\rightarrow\infty}(\EQ)^{\Ne n}$. Intuitively,
provided we choose $\epsilon$ to decrease fast enough, this should
converge to ${\Id}_R$. More precisely, we can write
$(\EQ)^\Ne={\Id}_R+\Ge$, where $\Ge$ is some map (need not be
CP) on $\RotE$ such that $\Vert\Ge\Vert <\epsilon$. Now consider the
map $(\EQ)^{\Ne n}=( {\Id}_R+\Ge)^n=\sum_{m=0}^n
\binom{n}{m}\Ge^m$, for $n\in\mathbb{N}$, which gives
\begin{equation}
\label{EQn}
\Vert(\EQ)^{\Ne n}-{\Id}_R\Vert\leq
\sum_{m=1}^n\binom{n}{m}\Vert\Ge^m\Vert\leq\epsilon(2^n-1).
\end{equation}
\noindent Let us choose $\epsilon=3^{-n}$ (actually,
$\epsilon=C_0^{-n}$ for any choice of $C_0>2$ works). Then taking the
limit $n\rightarrow\infty$ of Eq. \eqref{EQn}, we conclude that
$\lim_{n\rightarrow\infty}(\EQ)^{\Ne n}={\Id}_R$.

From this, we see that Eq. \eqref{eq:EinfA} can be rewritten as
\begin{equation}
\lim_{n\rightarrow\infty}\cE^n(A)=\sum_la_lX_l\quad\in\RotE.
\end{equation}
\noindent 
Therefore, $\Einf\equiv\lim_{n\rightarrow\infty}\cE^{n\Ne}$ (with
$\epsilon$ depending on $n$ as above) is the projection onto
$\RotE$. Since a unitarily noiseless code is preserved under any power
of $\cE$, it must be preserved under $\Einf$, which gives the desired
isometry condition.
\end{proof}

Note that $\Einf$ is CPTP simply because $\cE$ is CPTP, and the set of
CPTP maps on a finite-dimensional Hilbert space is closed under
composition.  Furthermore, it projects every operator onto the span of
the rotating points of $\cE$. Observe that $\RotE$ is precisely the
set of fixed points of $\Einf$.

\end{appendix}

%******************************************************************
%******************************************************************
\bibliographystyle{apsrev}

\end{document}